\documentclass[twocolumn,nofootinbib,superscriptaddress,aps,prd,longbibliography]{revtex4-1}
\usepackage[utf8]{inputenc}
\usepackage{amsfonts}
\usepackage{amsmath}
\usepackage{amsthm}
\usepackage{braket}
\usepackage{graphicx}
\usepackage{hyperref}
\usepackage[dvipsnames]{xcolor}
\usepackage{amssymb}
\usepackage{dsfont}
\usepackage{verbatim}
\usepackage{tikz,lipsum,lmodern}
\usepackage[most]{tcolorbox}
\usepackage[caption=false]{subfig} 

\newcommand{\zz}{\mathbb Z}

\newtheorem{proposition}{Proposition}

\newtheorem{lemma}{Lemma}

\newtheorem*{question*}{Question}

\usepackage[normalem]{ulem}

\graphicspath{{./Figures/}}

\newcommand{\groupone}{(\text{rg})}
\newcommand{\grouptwo}{(\text{rgb})}

\newcommand{\grprgb}{(\text{rgb})}

\newcommand{\grouponearg}{\text{rg}}
\newcommand{\grouptwoarg}{\text{rgb}}

\newcommand{\fr}{\textcolor{red!80!black}{\psi_\text{r}}}
\newcommand{\fg}{\textcolor{green!80!black}{\psi_\text{g}}}
\newcommand{\fb}{\textcolor{blue!80!black}{\psi_\text{b}}}
\newcommand{\condensible}{\mathcal{C}}

\newcommand{\lblr}{\textcolor{red!80!black}{\text{r}}}
\newcommand{\lblg}{\textcolor{green!80!black}{\text{g}}}
\newcommand{\lblb}{\textcolor{blue!80!black}{\text{b}}}

\begin{document}

\title{3-Fermion topological quantum computation}
\author{Sam Roberts}%
\affiliation{Centre for Engineered Quantum Systems, School of Physics, The University of Sydney, Sydney, NSW 2006, Australia}%
\author{Dominic J. Williamson}%
\affiliation{Stanford Institute for Theoretical Physics, Stanford University, Stanford, CA 94305, USA}
\date{\today}   

\begin{abstract}
We present a scheme for universal topological quantum computation based on Clifford complete braiding and fusion of symmetry defects in the 3-Fermion anyon theory, supplemented with magic state injection. 
We formulate a fault-tolerant measurement-based realisation of this computational scheme on the lattice using ground states of the Walker--Wang model for the 3-Fermion anyon theory with symmetry defects. 
The Walker--Wang measurement-based topological quantum computation paradigm that we introduce provides a general construction of computational resource states with thermally stable symmetry-protected topological order. 
We also demonstrate how symmetry defects of the 3-Fermion anyon theory can be realized in a 2D subsystem code due to Bomb\'{i}n -- making contact with an alternative implementation of the 3-Fermion defect computation scheme via code deformations. 
\end{abstract}

\maketitle

\section{Introduction}
Topological quantum computation (TQC) is currently the most promising approach to scalable, fault-tolerant quantum computation. In recent years, the focus has been on TQC with Kitaev's toric code~\cite{kitaev2003fault}, due to it's high threshold to noise~\cite{dennis2002topological,wang2003confinement}, and amenability to planar architectures with nearest neighbour interactions. 
To encode and manipulate quantum information in the toric code, a variety of techniques drawn from condensed matter contexts have been utilised. In particular, some of the efficient approaches for TQC with the toric code rely on creating and manipulating gapped-boundaries, symmetry defects and anyons of the underlying topological phase of matter~\cite{Rau06,TopoClusterComp,Raussendorf07,bombin2009quantum,bombin2010topologicaltwist,koenig2010quantum,barkeshli2013twist,teo2014braiding,teo2014unconventional,yoder2017surface,brown2017poking,khan2017fermion,lavasani2018low, zhu2020instantaneous,zhu2020universal,webster2020fault,barkeshli2019symmetry, bombin2021logical, ellison2022pauli, ellison2022pauli2}. Many of these approaches were discovered within the framework of measurement-based quantum computation (MBQC), which provides a natural space-time perspective to understand fault-tolerant quantum computations, as well as a deep connection to many-body physics through the \textit{computational phases of matter} paradigm~\cite{doherty2009identifying,raussendorf2019computationally}. Within the framework of MBQC, there have been many developments since the topological cluster state scheme based on the toric code~\cite{Rau06,TopoClusterComp,Raussendorf07}, including foliated schemes~\cite{bolt2016foliated}, non-foliated schemes~\cite{nickerson2018measurement,newman2020generating}, schemes based on 2d stabilizer codes~\cite{brown2020universal,lee2022universal}, approaches to implement higher-dimensional codes and protocols in 2D architectures~\cite{bombin20182d}, as well as the related construction of fusion-based quantum computation~\cite{bartolucci2021fusion,bombin2023faulttolerant}. 
Despite great advances, the overheads for universal fault-tolerant quantum computation remain a formidable challenge. 
It is therefore important to analyse the potential of TQC in a broad range of topological phases of matter, and attempt to find new computational substrates that require fewer quantum resources to execute fault-tolerant quantum computation. 

In this work we present an approach to TQC for more general anyon theories based on the Walker--Wang models~\cite{walker20123+}. This provides a rich class of spin-lattice models in three-dimensions whose boundaries can naturally be used to topologically encode quantum information. 
The two-dimensional boundary phases of Walker--Wang models 
accommodate a richer set of possibilities than stand-alone two-dimensional topological phases realized by commuting projector codes~\cite{von2013three,burnell2013exactly}. 
The Walker--Wang construction prescribes a Hamiltonian for a given input (degenerate) anyon theory, whose ground-states can be interpreted as a superposition over all valid worldlines of the underlying anyons. Focusing on a particular instance of the Walker--Wang model~\cite{burnell2013exactly} based on the 3-Fermion anyon theory (\textbf{3F} theory)~\cite{rowell2009classification,teo2015theory,bombin2009interacting,bombin2012universal}, we show that that the associated ground states can be utilised for fault-tolerant MBQC~\cite{raussendorf2001one,raussendorf2003measurement,van2006universal,Rau06,TopoClusterComp,Raussendorf07,brown2020universal} via a scheme based on the braiding and fusion of lattice defects constructed from the symmetries of the underlying anyon theory. 
The resource states required for the computation can be prepared with a Clifford circuit acting on a 2D grid with only nearest-neighbour interactions, and thus the architectural requirements for this approach are qualitatively similar to that of the widely pursued surface code schemes.
The Walker--Wang MBQC paradigm that we introduce provides a general framework for finding fault-tolerant resource states for universal computation. For example, we show that the well-known topological cluster state scheme for MBQC of Ref.~\cite{Rau06} is produced when the toric code anyon theory is used as input to the Walker--Wang construction. Therefore, our approach provides a generalization of the topological cluster-state scheme (which is based on the toric code anyon theory) to general abelian anyon models.

The \textbf{3F} theory is an interesting, and nontrivial example of the power of this framework. Owing to the rich set of symmetries of the \textbf{3F} theory, we find a universal scheme for TQC where all Clifford gates can be fault-tolerantly implemented and magic states can be noisily prepared and distilled~\cite{bravyi2005universal}. In particular, the full Clifford group in this scheme can be obtained by braiding symmetry twist defects. This is in contrast to the 2D toric code, where only a subgroup of Clifford operators can be achieved in by braiding symmetry twist defects (when using qubit encodings with fixed charge parity). We remark that this improved computational capability is derived from the symmetries of the anyon theory ($S_3$ for the \textbf{3F} theory, and $\zz_2$ for the toric code), as both the toric code and \textbf{3F} anyon theories consist of four anyons.

The \textbf{3F} Walker--Wang model --~and consequently the TQC scheme that is based on it~-- is intrinsically three-dimensional, as there is no commuting projector (e.g., stabilizer) code in two dimensions that realises the \textbf{3F} anyon theory~\cite{haah2018nontrivial,burnell2013exactly}. 
As such, this TQC scheme is outside the paradigm of operations on a 2D stabilizer code, 
and provides an important stepping stone towards understanding what is possible in general, higher-dimensional, topological phases. 
We remark, however, that it remains possible to embed our scheme into an extended nonchiral anyon theory that can be implemented in a 2D stabilizer model (such as the color code). We emphasize that the textbf{3F} theory is just one compelling example, and we expect further interesting examples to exist. 

Further connecting to the paradigm of \textit{computational phases of matter}, we ground our computational framework in the context of symmetry-protected topological (SPT) phases of matter. In particular, we explore the relationship between the fault-tolerance properties of our MBQC scheme and the underlying 1-form symmetry-protected topological order of the Walker--Wang resource state. While the 3D topological cluster state (of Ref.~\cite{Rau06}) has the same $\zz_2^2$ 1-form symmetries as the \textbf{3F} Walker--Wang ground state, they belong to distinct SPT phases.
These examples provide steps toward a more general understanding of fault-tolerant, computationally universal phases of matter~\cite{DBcompPhases,Miy10, else2012symmetry,else2012symmetryPRL,NWMBQC,miller2016hierarchy,roberts2017symmetry,bartlett2017robust,wei2017universal,raussendorf2019computationally,roberts2019symmetry,devakul2018universal,Stephen2018computationally,Daniel2019,daniel2020quantum}. 

Finally, we find another setting for the implementation of our computation scheme by demonstrating how symmetry defects can be introduced into the 2D subsystem color code of Bomb\'{i}n~\cite{bombin2010topologicalsubsystem,bombin2009interacting,Bombin_2011}, which supports a \textbf{3F} 1-form symmetry and has been argued to support a \textbf{3F} anyon phase. 
By demonstrating how the symmetries of the emergent anyons are represented by lattice symmetries, we make contact with an alternative formulation of the \textbf{3F} TQC scheme based on deformation of a subsystem code in (2{+}1)D~\cite{Bombin_2011} -- this may be of practical advantage for 2D architectures where 2-body measurements are preferred. Our construction of symmetry defects in this subsystem code may be of independent interest. By taking a certain limit of this model, our computational scheme embeds into a subtheory of the anyons and defects supported by Bomb\'{i}n's color code~\cite{bombin2006topological,bombin2009interacting,teo2014unconventional,yoshida2015topological,kesselring2018boundaries}.

\textit{Organisation.} In Sec.~\ref{sec3FPreliminaries} we review the \textbf{3F} anyon theory and its symmetries. In Sec.~\ref{sec3FTQCScheme} we present an abstract TQC scheme based on the symmetries of the \textbf{3F} theory. We show how to encode in symmetry defects, and how to perform a full set of Clifford gates along with state preparation by braiding and fusing them. In Sec.~\ref{sec3FWW} we review the \textbf{3F} Walker--Wang model, and then show how the symmetry defects and TQC scheme can be realized in the \textbf{3F} Walker--Wang Hamiltonian. In Sec.~\ref{sec3FMBQC} we show that the \textbf{3F} Walker--Wang model and associated symmetry defects can be used as a resource for fault-tolerant measurement-based quantum computation. We begin by reviewing MBQC based on the 3D topological cluster state~\cite{Rau06} and recasting it in the Walker--Wang MBQC paradigm. We present the qualitatively similar architectural requirements for the \textbf{3F} TQC scheme to that of currently pursued approaches. We also discuss the two models in the context of 1-form SPT phases. In Sec.~\ref{sec3FSubsystemCode} we show how the defects can be implemented in a 2D subsystem code, offering an alternative computation scheme based on code deformation. We conclude with a discussion and outlook in Sec.~\ref{sec3FDiscussion}.

\section{3-Fermion anyon theory preliminaries}\label{sec3FPreliminaries}

In this section we review the \textbf{3F} anyon theory, its symmetries and the associated symmetry domain wall and twist defects. 
We describe the fusion rules of the twists, including which anyons can condense on the twist defects.
The 3F model (also known as the $\text{SO}(8)_1$ theory) and its symmetries has been studied in Refs.~\cite{rowell2009classification,khan2014anyonic,teo2015theory}.
Closely related anomalous 1-form symmetry sectors, that do not necessarily correspond to anyons in a gapped phase, have also been studied in a subsystem code due to Bomb\'{i}n \cite{bombin2010topologicalsubsystem,bombin2009interacting,Bombin_2011}.

\subsection{Anyon theory}

The \textbf{3F} anyon theory describes superselection sectors $\{ 1,  \fr,\fg,\fb\}$  with $\mathbb{Z}_2\times\mathbb{Z}_2$ fusion rules
\begin{align}
\psi_\alpha \times \psi_\alpha = 1 \, ,
&&
 \fr \times \fg =   \fb \, ,
\end{align}
where $\alpha = \text{r,g,b},$ 
and modular matrices 
\begin{align}\label{eq3FModularMatrices}
S = \begin{pmatrix}
1 & 1 & 1 & 1 \\
1 & 1 & -1 & -1 \\
1 & -1 & 1 & -1 \\
1 & -1 & -1 & 1 
\end{pmatrix}
 \, ,
 &&
 T=  \begin{pmatrix}
1 & & &  \\
& -1 &  &  \\
&  & -1 & \\
 &  &  & -1 
\end{pmatrix}
\, .
\end{align}
The above $S$ matrix matches the one for the anyonic excitations in the toric code, but the topological spins in the $T$ matrix differ as $\fr,\fg,\fb$ are all fermions. 
These modular matrices suffice to specify the gauge invariant braiding data of the theory~\cite{WALL1963}, while the $F$ symbols are trivial. 

This anyon theory is \textit{chiral} in the sense that it is not consistent with a gapped boundary to vacuum~\cite{Moore1989,kitaev2006anyons}. 
Using the well-known relation 
\begin{align}
    \frac{1}{\mathcal{D}} \sum_{a\in \mathcal{C}} d_a^2 \theta_a = e^{2\pi i c_- / 8} \, , 
\end{align}
where $d_a$ is the quantum dimension and $\theta_a=T_{aa}$ is the topological spin of anyon $a$, ${\mathcal{D}^2=\sum_a d_a^2}$ defines the total quantum dimension, and $c_-$ is the chiral central charge. 
In the \textbf{3F} theory $d_a=1$ for all anyons as they are abelian, hence $\mathcal{D}=2$, we also have that $\theta_a=-1$ for all anyons besides the vacuum. Hence we find that the chiral central charge must take the value {$c_-=4$ mod 8}.

\subsection{Symmetry-enrichment}

\begin{figure}[t]%
	\centering
	\includegraphics[width=0.7\linewidth]{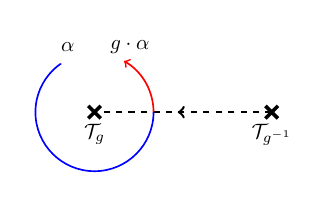}
	\caption{A fermion $\alpha \in \mathcal{C}$ is transformed by the symmetry group to $g\cdot \alpha \in \mathcal{C}$ under counter-clockwise braid with a twist $\mathcal{T}_g$. }
	\label{figSymmetryAction}
\end{figure}

The \textbf{3F} theory has an $S_3$ group of global symmetries corresponding to arbitrary permutations of the three fermion species, all of which leave the gauge invariant data of the theory invariant. 
We denote the group action on the three fermion types r,g,b, using cycle notation
\begin{align}
    S_3 \cong \{ (),(\text{rg}),(\text{gb}),(\text{rb}),(\text{rgb}),(\text{rbg}) \} \, ,
\end{align}
with the usual composition, e.g., ${(\text{rg})\cdot(\text{gb})=(\text{rgb})}$. 
The action on the anyons is then given by ${g\cdot 1=1,} \, { g\cdot \psi_c=\psi_{g\cdot c}}$. 

Restricting the action of a global symmetry to a subregion creates codimension-1 invertible domain walls~\cite{barkeshli2019symmetry}. 
These codimension-1 invertible domain walls are labelled by the nontrivial group elements. 
The codimension-2 topological symmetry twist defects that can appear at the open end of a terminated domain wall are labelled by their eigenvalues under the string operators for any fermions that are fixed by the action of the corresponding group element. 
Hence there are two distinct symmetry defects of quantum dimension $\sqrt{2}$ for each 2-cycle permutation which we label 
$\mathcal{T}_{(\text{rg})}^\pm,\mathcal{T}_{(\text{gb})}^\pm,\mathcal{T}_{(\text{rb})}^\pm,$ and there is only a single symmetry defect of quantum dimension $2$ for each of the 3-cycles which we label 
$\mathcal{T}_{(\text{rgb})}$ and $\mathcal{T}_{(\text{rbg})}$. Where we have utilized the fact that the total quantum dimension of each symmetry defect sector matches the trivial sector, consisting of only the anyons, and that the $\mathcal{T}_{(cc')}^\pm$ defects are related by fusing in either of the fermions $\psi_c,\psi_{c'}$ that are permuted by the action of the domain wall. 

The twist defect sectors of the full symmetry-enriched theory are then given by 
\begin{align}
    &\mathcal{C}_{S_3} = \{1,\psi_{\text{r}},\psi_{\text{g}},\psi_{\text{b}}\}
    \oplus \{\mathcal{T}_{(\text{rg})}^+,\mathcal{T}_{(\text{rg})}^-\} \oplus~ \nonumber \\
    &~ \{\mathcal{T}_{(\text{gb})}^+,\mathcal{T}_{(\text{gb})}^-\} \oplus \{\mathcal{T}_{(\text{rb})}^+,\mathcal{T}_{(\text{rb})}^-\} \oplus \{\mathcal{T}_{(\text{rgb})} \} \oplus \{ \mathcal{T}_{(\text{rbg})} \} \, ,
\end{align}
and the additional fusion rules for the defects are 
\begin{align}
    \mathcal{T}_{(cc')}^\pm \times \mathcal{T}_{(cc')}^\pm &= 1 + \psi_{c''}
     \\
    \mathcal{T}_{(cc')}^\pm \times \mathcal{T}_{(cc')}^\mp &= \psi_{c'}+\psi_{c''} 
    \\
    \mathcal{T}_{(\text{rgb})} \times \mathcal{T}_{(\text{rbg})} &= 1+\psi_r+\psi_g+\psi_b
\, , 
\end{align}
for $c\neq c' \neq c'' \neq c$,  and
\begin{align}
\mathcal{T}_{(cc')}^\pm \times \mathcal{T}_{(\text{rgb})^i} &= \mathcal{T}_{(cc')\cdot(\text{rgb})^i}^+ + \mathcal{T}_{(cc')\cdot(\text{rgb})^i}^-
\\
\mathcal{T}_{(\text{rgb})^{i}} \times \mathcal{T}_{(\text{rgb})^{i}} &= 2 \mathcal{T}_{(\text{rgb})^{-i}}
\, ,
\end{align}
for $i=\pm 1,$ and the related rules given by cycling the legs around a fusion vertex. 

These fusion rules imply what anyon types $C_{g}$ can condense on the $\mathcal{T}_g$ defects (the $\pm$ superscript makes no difference) as follows: 
\begin{align}
    \condensible_{(c c')} &= \{ 1,\psi_{c''}\} \, ,
    \\
    \condensible_{(c c' c'')} &= \{ 1,\fr,\fg,\fb\} \, ,
\end{align}
where $c\neq c'\neq c''\neq c$. 

We remark that the fusion algebra of each non-abelian $\mathcal{T}_{(cc')}^\pm$ twist defect with itself is equivalent to that of an Ising anyon or Majorana zero mode, reminiscent of the electromagnetic duality twist defect in the toric code~\cite{bombin2010topologicaltwist}.
A full description of the G-crossed braided fusion category~\cite{barkeshli2019symmetry} describing this symmetry-enriched defect theory is not needed for the purposes of this paper, as all relevant processes can be calculated using techniques from the stabilizer formalism. This theory has been studied previously, it is known to be anomaly free and in particular the theory that results from gauging the full symmetry group has been calculated~\cite{barkeshli2019symmetry,teo2015theory,Cui2016}.

We remark that in the following sections, for any 2-cycle $g\in S_3$ we define $\mathcal{T}_g$ (i.e., without a superscript) to be equal to $\mathcal{T}_g^+$. 
In particular, we do not make explicit use of $\mathcal{T}_g^{-}$ to encode logical information (although they may arise due to physical errors). 

\section{\textbf{3F} defect computation scheme}\label{sec3FTQCScheme}

\begin{figure}[t]%
	\centering
	\includegraphics[width=0.43\linewidth]{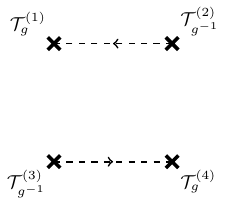}
	\caption{The elementary twist-defect configuration for encoding quantum information. One or two logical qubits are encoded if $g\in S_3$ is a 2-cycle (e.g., $\groupone$) or 3-cycle (e.g., $\grouptwo$), respectively.}
	\label{figBaseEncodings}
\end{figure}

In this section we demonstrate how to encode and process logical information using symmetry defects of the \textbf{3F} theory. Our scheme is applicable to any spin lattice model that supports \textbf{3F} topological order (possibly as a subtheory). Here we describe the scheme at the abstract level of an anyon theory with symmetry defects, with the microscopic details abstracted away.  
In the following sections we demonstrate how to realise our scheme via MBQC using a Walker--Wang model and in the 2D subsystem color code of Bomb\'{i}n~\cite{Bombin_2011,bombin2009interacting}. 

Our computational scheme is based on implementing a complete set of fault-tolerant Clifford operations using topologically protected processes -- which are naturally fault-tolerant to local noise, provided the twists remain well separated -- along with the preparation of noisy magic states. By Clifford operations we mean the full set of Clifford gates (the unitaries that normalise the Pauli group), along with single qubit Pauli preparations and measurements. The noisy magic states can be distilled to arbitrary accuracy using a post-selected Clifford circuit (provided the error rates are sufficiently small)~\cite{bravyi2005universal}. 
We remark that the schemes we present are by no means optimal, and given a compilation scheme and architecture, the overheads are ripe for improvement. 

The goal of this section is to prove Prop.~\ref{prop3FCliffordUniversality} -- the Clifford universality of \textbf{3F} defect theory -- which along with noisy magic state preparations offers a universal scheme for fault-tolerant quantum computation. We prove this proposition by breaking an arbitrary space-time configuration of domain walls and twists into smaller components that directly implement individual Clifford operations that generate the Clifford group and allow for Pauli preparations and measurements.
We begin by introducing defect encodings. 

\subsection{Encoding in symmetry defects}\label{sec3FEncodings}

By nature of their ability to condense anyonic excitations, symmetry defects are topological objects and information can be encoded in them. To understand such encodings, we consider a two-dimensional plane upon which anyonic charges -- in our case, fermions in $\mathcal{C}$ -- and symmetry defects may reside. This setting is representative of the behaviour of anyons that arise as excitations on two-dimensional topologically ordered phases -- in our case the fermions appear as excitations on the boundary of the three-dimensional Walker--Wang model as well as in the low energy theory of a 2D subsystem code Hamiltonian. Processes that involve moving, braiding and fusing of anyons can be realized on the lattice by certain string operators. Such string operators can also transfer anyonic charge to (and between) twist defects, thereby changing their topological charge.

For a given configuration of twist defects $\{ \mathcal{T}_{g_i}^{(i)}, ~|~ i \in \{1,\ldots,N\}, g_i \in S_3\}$, we can encode a quantum state in the joint fermionic charge of $g$-neutral subsets of them $\mathcal{I}\subseteq\{1,\ldots,N\}$. By $g$-neutral, we mean that the subset of twist defects $\{\mathcal{T}_{g_i}^{(i)}~|~i \in \mathcal{I}\}$ must satisfy $\prod_{i\in \mathcal{I}}g_i = 1$. As the subsets are $g$-neutral, upon their fusion we are left with a fermionic charge $c\in\mathcal{C}$. These possible post-fusion charge states give us a basis for our encoded state space, and the dimension of the logical state-space depends on the quantum dimension of the defects. 

\textbf{$g$-encodings:}  To be more concrete we fix a twist configuration that acts as the fundamental encoding unit, known as the $g$-encoding $\mathcal{E}_g$, where ${1 \neq g\in S_3}$. In the following, all twists are of the $+$-type, where relevant. The encoding is defined by two twist pairs ${\mathcal{E}_g = \{\mathcal{T}_{g}^{(1)},\mathcal{T}_{g^{-1}}^{(2)},\mathcal{T}_{g^{-1}}^{(3)},\mathcal{T}_{g}^{(4)}\}}$ for $g\in S_3$ with vacuum total charge, as depicted in Fig.~\ref{figBaseEncodings}. The computational basis is defined by the fusion space of $\mathcal{T}_{g}^{(1)}$ and $\mathcal{T}_{g^{-1}}^{(2)}$: when $g$ is a 2-cycle, the two pairs encode a single qubit, and when $g$ is a 3-cycle the two pairs encode two qubits. This degeneracy follows from the fusion space of the twists 
\begin{align}
    \mathcal{T}_{\groupone} \times \mathcal{T}_{\groupone} &= 1 + \fb, \\
    \mathcal{T}_{\grouptwo} \times \mathcal{T}_{(\text{rbg})} &= 1 + \fr + \fg + \fb,
\end{align}
along with the constraint that all four twists must fuse to the vacuum containing no charge. 
For instance, when $g=\groupone$ the $\ket{\overline{0}}$ state corresponds to the fusion outcome $1\in \mathcal{C}$, and the $\ket{\overline{1}}$ state corresponds to outcome $\fb \in \mathcal{C}$. 

\begin{figure}[t]%
	\centering
	\includegraphics[width=0.95\linewidth]{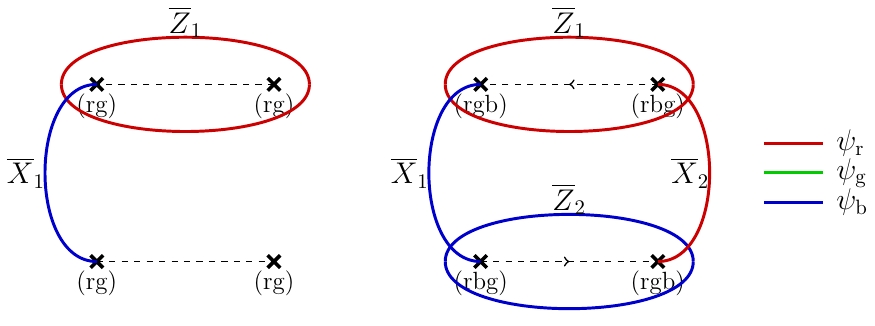} 
	\caption{Representative fermionic string operators for logical Pauli operators for a $g$-encoding. (left) A single qubit is encoded in four twists defined by $\groupone\in S_3$. Note also in this case that the orientation has been removed from the domain wall as $\groupone^{-1} = \groupone$. Similar representative logical operators for twist defects based on the other 2-cycles $g\in S_3$ can be obtained by suitably permuting the fermionic string operator types. (right) Two qubits are encoded in four twists defined by $\grouptwo,(\text{rbg}) \in S_3$. }
	\label{figQubitBasis}
\end{figure}

We remark that the exact location of the domain wall is not important in the encoding of Fig.~\ref{figBaseEncodings}; only their end points matter, as the action in Fig.~\ref{figSymmetryAction} is invariant under deformations of the domain wall. To encode qubits, one can choose any domain wall configuration with the same end points as the twist defects in Fig.~\ref{figBaseEncodings}.

The total fermionic charge of a subset of $g$-neutral defects can be detected without fusing the twists together, by instead braiding various fermionic charges around them. Such a process can be represented by a string operator (also known as a Wilson loop), and this loop can be used to measure the charge within the defects. Similarly, one can change the charge on each twist by condensing fermions into them, which is also represented by a string operator running between pairs of twists. 

The string operators that represent Pauli logical operators $\overline{X}$, $\overline{Z}$ acting on the encoded qubits are represented in Fig.~\ref{figQubitBasis} -- they can be understood as transferring and measuring fermionic charge between different defect pairs. Such operators must anticommute based on the mutual semionic braiding statistics of the fermions they transport (i.e., braiding one fermion around another introduces a $-1$ phase). It is often convenient to utilise other representative logical operators. For instance, when $g$ is a 2-cycle (e.g., $g=\groupone$), one can use either the $\fr$ or $\fg$ loops to measure the charge and hence define the logical $\overline{Z}$ operator. This follows from the fact that that an $\fb$ Wilson loop enclosing $\mathcal{T}_{\groupone}^{(1)}$ and $\mathcal{T}_{\groupone}^{(2)}$ acts as the logical identity, and swaps $\fr$ and $\fg$ loops upon fusion. In addition, logical $\overline{X}$ can be represented as a loop operator as per Fig.~\ref{fig3F2DEncodingsIsotope}.

\begin{figure}[t]%
	\centering
	\includegraphics[width=0.98\linewidth]{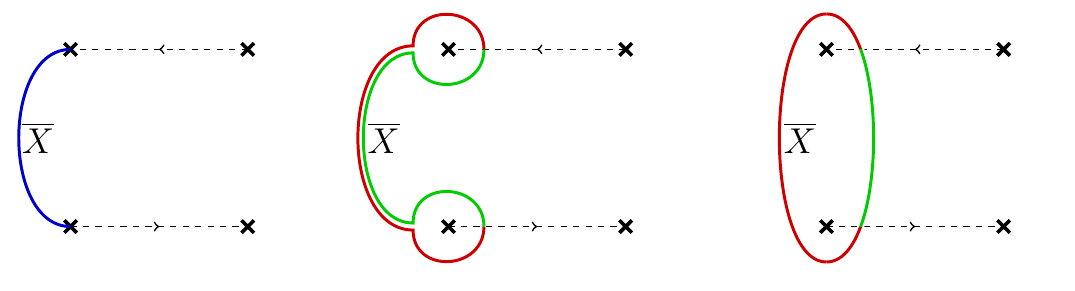}
	\caption{Equivalence between different representative fermionic string operators for logical Pauli $\overline{X}$ operators for the two-twist-pair encoding for $g$-encodings -- in this case $g=\groupone$. They can be verified by the fusion rule for $\fr \times \fg = \fb$ along with the fact that $\mathcal{T}_{\groupone}$ can condense $\fb$ fermions. }
	\label{fig3F2DEncodingsIsotope}
\end{figure}

More efficient encodings are possible. For instance, one can encode $N$ ($2N$) logical qubits into $(2N+2)$ $\text{2-cycle}$ (3-cycle) twists on the sphere, following for example Ref.~\cite{barkeshli2013twist}. Additionally, due to the rich symmetry defect theory of \textbf{3F} other encodings are possible, including a trijunction encoding which is outlined in App.~\ref{appOtherEncodings}.

\subsection{Gates by braiding defects}\label{sec3FGates}
We now show how to achieve encoded operations (gates, preparations and measurements) on our defect-qubits. In order to implement these operations, we braid twists to achieve gates, and fuse them to perform measurements. To understand such processes, we describe the locations of twists in (2{+}1)-dimensions. In (2{+}1)-dimensions, the twists -- which are codimension-2 objects -- can be thought of as worldlines. The domain walls -- which are codimension-1 objects -- can be thought of as world-sheets. 

\begin{figure}[t]%
	\centering
	\includegraphics[width=0.28\linewidth]{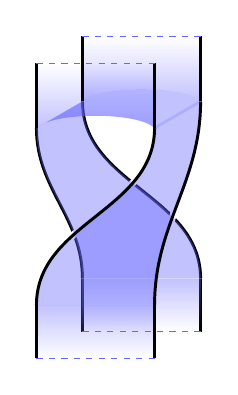} \qquad
	\includegraphics[width=0.28\linewidth]{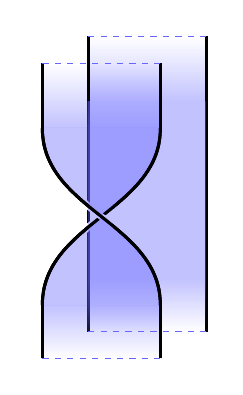} 
	\caption{One qubit gates for the defect encoding $\mathcal{E}_{\groupone}$. Time moves upwards. (a) The Hadamard gate cyclically permutes the four twist defects. A domain wall plane is inserted to return the encoding to its standard form. (b) The $S$ gate consists of the exchange of $\mathcal{T}_{\groupone}^{(3)}$ and $\mathcal{T}_{\groupone}^{(4)}$. The same gates work for a $g$-encoding with $g$ a 2-cycle -- in this case the orientation of the surface does not matter and is not depicted. } 
	\label{fig3FHandP}
\end{figure}

\begin{lemma}\label{lemSingleQubit}
Braiding the twists of a $g$-encoding $\mathcal{E}_g$ with $g$ a 2-cycle generates the single qubit Clifford group $\langle H, S \rangle$ where $H$ is the Hadamard and $S$ is the phase gate. 
\end{lemma}
\begin{proof}
The proof is presented in App.~\ref{appProofOfGates}. 
\end{proof}

We remark that in the case that $g$ is a 3-cycle, each $\mathcal{E}_g$ encodes 2 logical qubits. Braiding in this case generates a subgroup of the Clifford group given by $\langle H_{(1)}H_{(2)}, S_{(1)}S_{(2)} \rangle$, where the subscript indexes the two logical qubits. 

The previous Lemma defines a generating set of single qubit braids. We present the space-time diagram for the Hadamard and $S$ gate braids in Fig.~\ref{fig3FHandP}. Such diagrams can be interpreted in terms of code deformations or in terms of measurement-based quantum computation. In the former, we depict the space-time location of twists and domain walls during a code deformation, wherein twists trace out (0{+}1)-dimensional worldlines, and domain walls trace out (1{+}1)-dimensional worldsheets. In the MBQC picture, we similarly depict the location of twists and domain walls which correspond to lattice defects within the resource state, as we show explicitly in the next section. As in the previous section, the exact location of domain wall worldsheets is not important and only the location of the twist worldlines matter -- and they must remain well-separated in order for logical errors to be suppressed from local noise processes.

For entangling gates we require encodings $\mathcal{E}_g$ and $\mathcal{E}_h$ with either $g\neq h$, or at least one of $g$, $h$ being a 3-cycle. 

\begin{lemma}\label{lemEntangling}
Braiding of twists from two encodings $\mathcal{E}_g$ and $\mathcal{E}_h$ generates entangling gates if and only if either $g \neq h$ or at least one of $g$, $h$ is a 3-cycle.
\end{lemma}

\begin{proof}
See App.~\ref{appProofOfGates}. 
\end{proof}

Similarly, we present the space-time diagram of the Controlled-$Z$ ($CZ$) gate between two qubits encoded within 2-cycle encodings $\mathcal{E}_{\groupone}$ and $\mathcal{E}_{(\text{rb})}$ in Fig.~\ref{fig3FDefectCZ}. If one wishes to implement an entangling gate between two $\groupone$-encoded qubits -- such as a $CZ$ gate -- one can utilize an $(\text{rb})$-encoded ancilla to achieve this, as shown by the circuit in Fig.~\ref{fig3FEntanglingCircuit} in App.~\ref{SecEntanglingCircuitViaAncilla}.

\begin{figure}[t]%
	\centering
	\includegraphics[width=0.45\linewidth]{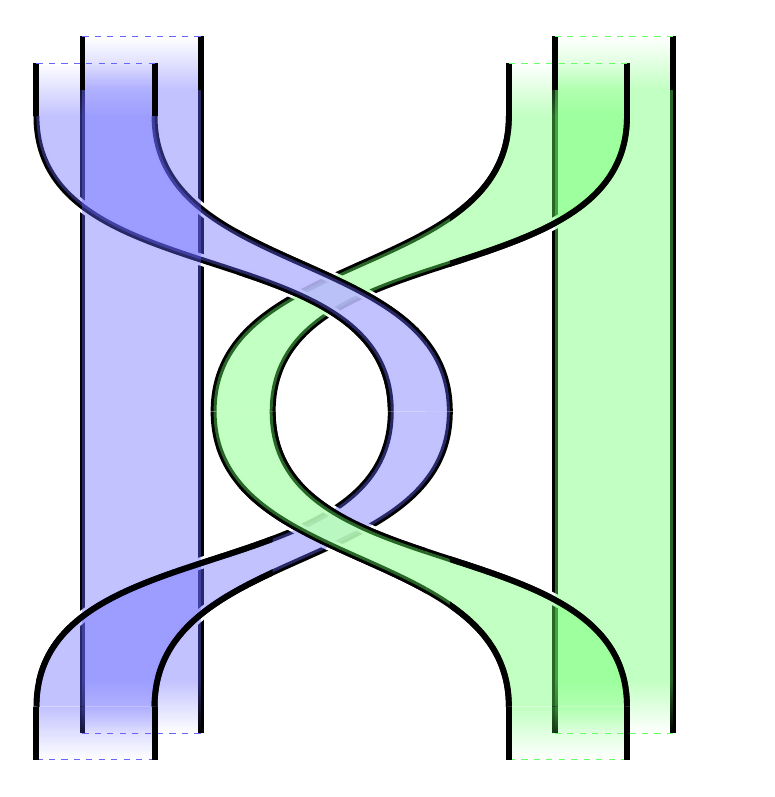} 
	\caption{Two qubit $CZ$ gate between pairs of qubits with $g$- and $h$- encodings with 2-cycles $g\neq h \in S_3$, e.g., $g=\groupone$ on the left and $h=(\text{rb})$ on the right. Time moves upwards. Domain walls are colored according to the fermion that they leave invariant.  In App.~\ref{SecEntanglingCircuitViaAncilla} we show how to generate entangling gates between two $\groupone$-encoded qubits.}
	\label{fig3FDefectCZ}
\end{figure}

We remark that if one implements the same operation as in Fig.~\ref{fig3FDefectCZ} using two $\grouptwo$-encodings (which encode 4-qubits) we obtain the operation $\overline{CZ}_{1,4}\overline{CZ}_{2,3}$, where qubits 1,2 belong to the left $\grouptwo$-encoding, and qubits 3,4 belong to the right $\grouptwo$-encoding.

One can understand the action of these gates by tracking representative logical operators through space-time. If the braid is implemented by code deformation, the logical mapping can be understood by tracking representative logical operators at each time slice through the space-time braid. In the context of MBQC, the observable that propagate logical operators through space-time are known as correlation surfaces -- they reveal correlations in the resource state that determine the logical action on the post-measured state (see for example \cite{Raussendorf07}). Correlation surfaces for each operation are determined in App.~\ref{appProofOfGates}.

We remark that these operations can be described purely in terms of the braid group acting on twists. This can be useful for topological compilation, where one can find more efficient representations of general Clifford operations. We define the braids that give rise to the Hadamard, $S$ gate and $CZ$ gate in App.~\ref{SecEntanglingCircuitViaAncilla}.

\subsection{Completing a universal gate set}
To complete the set of Clifford operations, we require Pauli basis measurements, which are obtained by fusing twists together. To obtain a universal set of operations we show how to prepare noisy magic states, which can then be distilled using Clifford operations -- allowing for fault-tolerant universality~\cite{bravyi2005universal}. 

While the 3-cycle encodings can be used, we focus on a universal scheme for quantum computation using $\groupone$-encoded qubits for logical qubits, and $(\text{rb})$-encoded qubits as ancillas to mediate entangling gates. 

\subsubsection{State preparation}\label{secMagicStatePreparation}

We now show how to perform topologically protected measurements in the $\overline{X}$ and $\overline{Z}$ basis, as well as preparations, which can be considered time-reversed measurements (and vice versa). To prepare a state in $\overline{X}$ or $\overline{Z}$ we must nucleate out the twists of $\mathcal{E}_g$ such that we know the definite (fermionic) charge of $\mathcal{T}_g^{(1)} \times \mathcal{T}_{g^{-1}}^{(3)}$ and $\mathcal{T}_g^{(1)} \times \mathcal{T}_{g^{-1}}^{(2)}$. These basis preparations are depicted in Fig.~\ref{preparations}. In the case that $g$ is a 3-cycle, both qubits are prepared in the same basis. This completes the set of Clifford operations.

\begin{figure}[t]%
	\centering
	\includegraphics[width=0.75\linewidth]{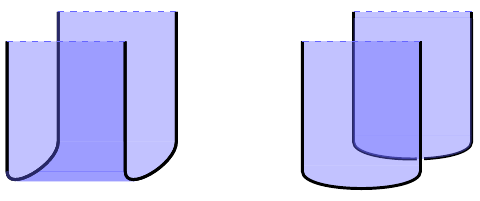}
	\caption{(left) Preparing $\overline{X}$ eigenstates. (right) Preparing $\overline{Z}$ eigenstates. We depict the operation for a $\groupone$-encoding. Time moves upwards. To prepare either $\overline{X}$ or $\overline{Z}$ eigenstates we need to prepare pairs of twists in definite charge states. This can be done by nucleating them out of vacuum so we know they fuse to the identity anyon (i.e., no charge). To obtain the respective measurements, we take the time-reverse diagram (i.e., $t\rightarrow -t$). This works identically for any $g\in S_3$, and we note that when $g$ is a 3-cycle, both encoded qubits are prepared (or measured) in the same basis.}
	\label{preparations}
\end{figure}

\begin{proposition}\label{prop3FCliffordUniversality}
(Clifford universality of \textbf{3F} defect theory). For any 2-cycles $g\neq h\in S_3$, any Clifford operation can be implemented on $g,h$-encoded qubits by braiding and fusion of twists.
\end{proposition}
\begin{proof}
An arbitrary Clifford operation is given by either a Clifford unitary -- which can be generated by Hadamard, phase and $CZ$ -- or by a single qubit Pauli preparation or measurement. All Clifford unitaries can be implemented by Lemmas \ref{lemSingleQubit}, \ref{lemEntangling}, and the circuit identity of Fig.~\ref{fig3FEntanglingCircuit} in App.~\ref{SecEntanglingCircuitViaAncilla}. This, along with the Pauli $X$ and $Z$ basis preparations and measurements, as demonstrated in App.~\ref{appProofOfGates} completes the proof.
\end{proof}

To complete a universal set of gates we consider preparation of noisy $T$-states $\ket{\overline{T}} = \frac{1}{\sqrt{2}}(\ket{\overline{0}}+e^{\frac{i\pi}{4}}\ket{\overline{1}})$. Such states can be distilled using post-selected Clifford circuits, and are sufficient to promote the Clifford gateset to universality~\cite{bravyi2005universal}. To prepare noisy $T$ states, we utilise a non-topological projection on the four twists in a $g$-encoding that are brought within a constant width neighbourhood. Here we consider $g$ a 2-cycle for simplicity. The $\ket{\overline{T}}$ state is the $+1$ eigenstate of $M_{\overline{T}}=(\overline{X} +\overline{Y})/\sqrt{2}$, and thus its preparation can be achieved by measuring the observable $M_{\overline{T}}$ and post selecting on the $+1$ outcome. To ensure such operations can achieved in a local way, the four twists of the $g$-encoding must be brought within a small neighbourhood to perform the (noisy) measurement $M_{\overline{T}}$, after which they can be separated. In the Walker--Wang resource states introduced in the following section, it is possible to bring the twists within a constant separation such that $M_{\overline{T}}$ is a constant-sized operator. (Note that we do not explore whether one can bring the twists to within a distance of one lattice spacing, such that the required logical action can be implemented with the single-qubit measurement $(X+Y)/\sqrt{2}$, as is possible in the surface code case~\cite{li2015magic,lodyga2015simple,singh2022high,bombin2022fault}.) Topologically, the magic state preparation is depicted in Fig.~\ref{figMagicState}. 

\begin{figure}[t]%
	\centering
	\includegraphics[width=0.99\linewidth]{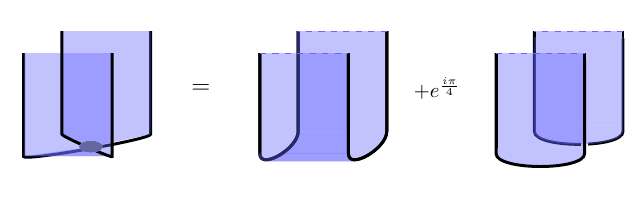} \\
	\caption{Non-topologically protected magic state preparation. Time moves upwards. Four twists are brought to close proximity such that a non-topological operation can be implemented (depicted by shaded neighbourhood around all four twists on the left-most figure) -- in this case to prepare $\ket{\overline{T}} = \frac{1}{\sqrt{2}}(\ket{\overline{0}}+e^{\frac{i\pi}{4}}\ket{\overline{1}})$. The precise nature of the non-topological projection depends on the lattice implementation. Topologically, the projection can be understood as giving rise to a superposition of a $\overline{X}$ and $\overline{Z}$ eigenstate preparations.}
	\label{figMagicState}
\end{figure}

\section{Walker--Wang realisation of \textbf{3F} 
computational resource states}\label{sec3FWW}

In order to implement the computational schemes of the previous section, we develop a framework for MBQC based on Walker--Wang resource states. In this section we introduce the \textbf{3F} Walker--Wang model of Ref.~\cite{burnell2013exactly} which provides the resource state for our computation scheme. We describe how the symmetries of the \textbf{3F} anyon theory can be lifted to a lattice representation, as symmetries of the \textbf{3F} Walker--Wang model, along with how to implement symmetry domain walls and twists based on them. While we focus on the \textbf{3F} anyon theory, the Walker--Wang construction, along with our computation scheme can be applied for general anyon theories. Indeed, the most well known example of fault-tolerant MBQC -- the topological cluster state model of Ref.~\cite{Rau06} -- is a special case of our construction, that arises when the toric code anyon theory is used as an input as described in Sec.~\ref{subsecComparisonToRauss}. We expect more exotic MBQC schemes can be found using this paradigm. 
However, for general non-abelian anyon theories efficiently accounting for the randomness of measurement outcomes is an open problem.

\subsection{Hilbert space and Hamiltonian}

The Walker--Wang model~\cite{walker20123+} extends the string--net model~\cite{levin2005string} to (3+1)D, defining a Hamiltonian for any braided anyon theory, including those with degenerate braiding~\cite{Moore1989,kitaev2006anyons}. 
The degrees of freedom of the model have a basis labelled by the anyon types of the input anyon theory.
The Hamiltonian is designed to energetically favour a ground state that is the superposition of all valid anyon worldline diagrams, weighted by their evaluation in the anyon theory. 
For modular, nondegenerate, braided anyon theories this model results in a bulk that has only topologically trivial excitations\footnote{There is a subtle sense in which the bulk of a Walker--Wang for a chiral modular anyon model is nontrivial~\cite{haah2018nontrivial,Haah2019,Haah2022a,Haah2022b,Shirley2022}.}.
A smooth boundary of the model supports (2+1)D boundary states with topological order corresponding to the input anyon theory. 
In this work, our focus is on a particular instance of the Walker--Wang model that is significantly simpler than the general case.  
For this reason we will not present further details about the general construction here.

We utilise the simplified \textbf{3F} Hamiltonian defined in Ref.~\cite{burnell2013exactly}. We begin by considering a cubic lattice $\mathcal{L}$ with periodic boundary conditions. (For the \textbf{3F} theory, we do not need to trivalently resolve the cubic lattice as is done for general Walker--Wang models). The Hilbert space is given by placing a pair of qubits on each 1-cell of the cubic lattice $\mathcal{L}$. We refer to each 1-cell as a \textit{site}. We label a basis for each site $i$ as $\ket{x_1 x_2}_i, x_1,x_2 \in \mathbb{Z}_2$. Pauli operators acting on the first (second) qubit of site $i$ are labelled by $\sigma_i^{\alpha}$ ($\tau_i^{\alpha}$), where $\alpha \in \{ X, Y, Z\}$.

Following Ref.~\cite{burnell2013exactly}, the \textbf{3F} Hamiltonian is defined in terms of vertex and plaquette operators
\begin{equation}\label{eq3FWWHamiltonian}
    H_{\textbf{3F}} = -\sum_{v\in V}( A_v^{(\fr)} + A_v^{(\fg)}) -\sum_{f \in F} (B_f^{(\fr)} + B_f^{(\fg)}),
\end{equation}
where the sum is over all vertices $V$ and plaquettes $F$ of the lattice, and
\begin{align}
    A_v^{(\fr)} &= \prod_{i \in \delta v} \sigma_i^{X},  &\quad B_f^{(\fr)} &= \sigma_{O_f}^X \sigma_{U_f}^X \tau_{U_f}^X \prod_{i \in \partial f}\sigma_i^Z, \label{eqHamTermse}\\
    A_v^{(\fg)} &= \prod_{i \in \delta v} \tau_i^{X}, &\quad B_f^{(\fg)} &= \sigma_{O_f}^X \tau_{O_f}^X \tau_{U_f}^X\prod_{i \in \partial f}\tau_i^Z.\label{eqHamTermsm}
\end{align}
Therein $\delta v$ consists of all edges that contain $v$ as a vertex, $\partial f$ consists of all edges belonging to the face $f$, and $O_f$ and $U_f$ are the unique edges determined by the plauqette $f$ as per Fig.~\ref{fig3FLatticeExample}. We also define the terms $A_v^{(\fb)} = A_v^{(\fr)} A_v^{(\fg)}$, $B_f^{(\fb)} = B_f^{(\fr)} B_f^{(\fg)}$, and one may add them to the Hamiltonian (with a negative sign) if desired. We remark that not all terms are independent, for example, taking products of plaquettes around a cube gives the product of a pair of vertex terms. 

\begin{figure*}[t]%
	\centering
	\includegraphics[width=0.4\linewidth]{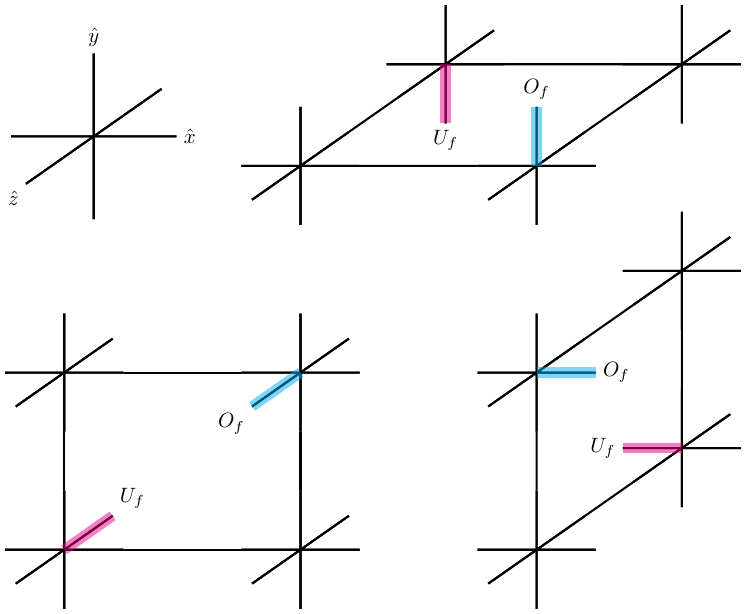} \hspace{1cm}
	\includegraphics[width=0.35\linewidth,trim= 0 -1.75cm 0 0]{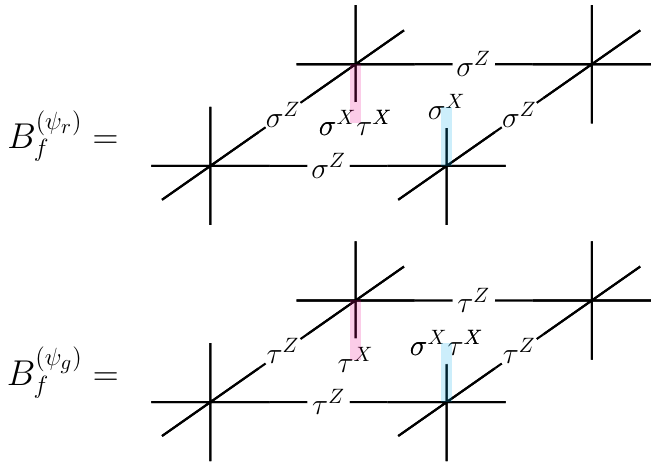}   
	\caption{(left) Special edges $O_f$ and $U_f$ for each plaquette orientation. The coordinate system is shown, with each edge of the lattice being length 1. (right) Example of Hamiltonian terms $B_f^{(\fr)}$ and $B_f^{(\fg)}$.}
	\label{fig3FLatticeExample}
\end{figure*}

On any closed manifold (i.e., without boundary), the ground state of $H_{\textbf{3F}}$ is unique~\cite{haah2018nontrivial}. In the Walker--Wang description, the ground state of $H_{\textbf{3F}}$ can be viewed as a weighted superposition over all valid anyonic worldlines, i.e., braided anyon diagrams that can be created from the vacuum via a sequence of local moves.
In particular, for each link, the basis of $\sigma_i^X$ and $\tau_i^X$ can be viewed as defining the presence or absence of fermionic $\fr$ and $\fg$ strings: $\ket{++}$ denotes the vacuum (identity anyon), $\ket{-+}$ denotes the presence of $\fr$, $\ket{+-}$ denotes the presence of $\fg$, and $\ket{--}$ denotes the presence of $\fb$. The $A_v^{(\fr)}$, $A_v^{(\fg)}$ terms generate a $\zz_2\times \zz_2$ 1-form symmetry, ensuring valid fusion rules at each vertex (i.e., $\zz_2\times \zz_2$ fermion conservation), while the $B_f^{(\fr)}$, $B_f^{(\fg)}$ ``fluctuation" terms ensure the ground-space is a superposition over all valid fermionic worldline configurations, with sign determined by the fermion braiding rules. Namely, the unnormalized ground state is 
\begin{equation}
    \ket{\psi_{\textbf{3F}}} = \sum_{c \in \mathcal{D}} \phi(c) \ket{c}, \quad \phi(c) = (-1)^{\text{linking}(c) + \text{writhe}(c)},
\end{equation}
where $\mathcal{D}$ is the set of all basis states corresponding to closed anyon diagrams with valid fusion rules that can be created from the vacuum, and $\text{linking}(c)$ ($\text{writhe}(c)$) is the linking number (writhe number) of the $\fr$ and $\fg$ fermion worldlines~\cite{walker20123+}.

\subsection{Symmetry of the \textbf{3F} Hamiltonian}
Recall that the \textbf{3F} theory has a symmetry $S_3 = \text{Aut}(\mathcal{C})$ with action on anyons given by $g \cdot 1 = 1$, $ g \cdot \psi_i = \psi_{g(i)}$, where $g(i)$ denotes the usual $S_3$ permutation action on $i \in \{\lblr, \lblg, \lblb\}.$ 
We now show that this symmetry can be lifted to a symmetry of the \textbf{3F} Walker--Wang model defined above. 

The symmetry contains an onsite and non-onsite part. Namely, write the symmetry $S(g)$ of the \textbf{3F} Hamiltonian as 
\begin{equation}
S(g) = V(g) U(g) \quad g \in S_3,
\end{equation}
where $U(g)$ is the onsite representation of $S_3$ and $V(g)$ is a locality preserving unitary (the deviation from onsiteness), which takes the form of a partial translation of qubits. If we write the basis for the 2 qubit space on each 1-cell as $\ket{\textbf{1}} := \ket{++}$, $\ket{\fr} := \ket{-+}$, $\ket{\fg} := \ket{+-}$, $\ket{\fb} := \ket{--}$, then the onsite part of the symmetry acts as a permutation of the three fermionic basis state on each site,
\begin{align}\label{eqFermionPermuation}
    g\cdot\ket{\psi_k}_i &= \ket{\psi_{g \cdot k}}_i, \quad g \in S_3, ~k \in \{\lblr,\lblg,\lblb\},
\end{align}
while preserving the vacuum, i.e, $g\cdot\ket{\textbf{1}}_i = \ket{\textbf{1}}_i$. This action can be represented by a Clifford unitary on each site.

The unitary, onsite representation $U$ of $S_3$ is defined by $U(g) = \otimes_i u_i(g)$, where we have generators
\begin{align}
u_i(\grouponearg) &= \text{SWAP}_{i_1, i_2}, \label{eqOnsiteSym1} \\
u_i(\grouptwoarg) &= \text{SWAP}_{i_1, i_2} \cdot \text{CNOT}_{i_1, i_2} . \label{eqOnsiteSym2}
\end{align}
The non-onsite part $V$ of the representation is generated by
\begin{align}
V(\grouponearg) &= T_{\tau}(v) \quad \text{with} \quad v = (1,1,1), \label{eqGaugeTrans1} \\
V(\grouptwoarg) &= I \label{eqGaugeTrans2},
\end{align}
where $T_{\tau}(v)$ is a partial translation operator acting on all $\tau$ qubits, shifting them in the $v = (x,y,z)$ direction, with coordinate basis defined in Fig.~\ref{fig3FLatticeExample}. 
Notationally, we use only single parentheses when explicit group elements appear as representation arguments, e.g., ${V(\groupone)\equiv V\groupone}$.

The partial translation operator has a well defined action on operators as a translation of their support. Namely, it can be defined factor-wise (with respect to the tensor product) for Pauli operators $\tau_u^{\alpha}$, $\sigma_u^{\alpha}$ at coordinate $u$, we have $T_{\tau}(v): \tau_u^{\alpha} \mapsto \tau_{u+v}^{\alpha}$, $\sigma_u^{\alpha} \mapsto \sigma_u^{\alpha}$, and extended by linearity. 

\begin{proposition}\label{prop3FHamSymmetry}
The unitary representation $S$ of $S_3$ defined by Eqs.~(\ref{eqOnsiteSym1}) -- (\ref{eqGaugeTrans2}) is a symmetry of the Hamiltonian $H_{\textbf{3F}}$.
\end{proposition}
That $u_i$ is indeed a representation of $S_3$ is verified in App.~\ref{AppProofOfSymmetryRep}. We note that commutation of the symmetry with the Hamiltonian is not strictly necessary here, since $H_{\textbf{3F}}$ is a stabilizer model it is sufficient to prove that the stabilizer group is preserved under the action of $S(g)$ $\forall g\in S_3$. This can be verified by direct computation and we provide the proof in App.~\ref{AppProofOfProp3FHamSymmetry}. We remark that $S(g)$ induces a permutation on the terms $B_f^{(\fr)}$, $B_f^{(\fg)}$, $B_f^{(\fb)}$, given by $S(g) B_f^{(f_i)} S^{-1}(g) = B_f^{(f_{g\cdot i})}$.

Thus only the 3-cycles have an onsite representation while the 2-cycles require a non-onsite partial translation. One can track the source of the non-onsiteness to the particular choice of gauge for the input data to the Walker--Wang construction -- namely the $R$ symbols -- to obtain the Hamiltonian $H_{\textbf{3F}}$. One can equally construct a Hamiltonian using the transformed data corresponding to the action of each symmetry element $g\in S_3$, all of which belong to the same phase and can be related by a locality preserving unitary. This additional locality preserving unitary is the origin of the non-onsite part of the symmetry. 
In general, applying a global symmetry to an anyon theory results in a transformation of the gauge-variant data~\cite{barkeshli2019symmetry}. The Walker--Wang model based on this transformed data is in the same topological phase as the original model, implying the existence of a locality-preserving unitary to bring the symmetry transformed Hamiltonian back to the original Hamiltonian. 
Combining the global symmetry transformation with this locality-preserving unitary promotes the global symmetry of the input anyon theory to a locality-preserving symmetry of the Walker--Wang Hamiltonian.  
We remark that for the \textbf{3F} theory, and more general anyon theories, one can construct a (non-stabilizer) Hamiltonian representative (using the symmetry-enriched anyon theory data), where the symmetry is onsite~\cite{shawnthesis,williamson2016hamiltonian} -- even for anomalous symmetries~\cite{bulmash2020absolute}. 

\subsubsection{Transforming the lattice}
For the \textbf{3F} Hamiltonian presented in Eq.~(\ref{eq3FWWHamiltonian}), the 3-cycles $\grouptwo, (\text{rbg}) \in S_3$ admit an onsite unitary representation, while the 2-cycles require a non-onsite (but nonetheless locality preserving) unitary. By transforming the lattice, we can express the symmetry action of the 2-cycles entirely as a translation. This simplifies the implementation of symmetry defects on the lattice.  Namely, consider the translation operator
\begin{equation}\label{eqLatticeTransformation}
    T(t), \quad t= \frac{1}{2}(1,1,1),
\end{equation}
that acts to translate all qubits in the $t$ direction (where again the $(x,y,z)$ coordinates are defined in Fig.~\ref{fig3FLatticeExample}). Consider the translation $T_{\tau}(v)$ of all of the $\tau$ qubits such that they are shifted to faces of the cubic lattice. On this new lattice, where there are qubits on every face ($\tau$ qubits) and every edge ($\sigma$ qubits) and the \textbf{3F} Hamiltonian consists of a term for each vertex $v$, edge $e$, face $f$ and volume $q$: 
\begin{align}
    \tilde{A}_v^{(\fr)} &= \prod_{e \in \delta v} \sigma_e^{X},  &\quad \tilde{B}_f^{(\fr)} &= \sigma_{O_f}^X \sigma_{U_f}^X \tau_{f}^X \prod_{e \in \partial f}\sigma_e^Z, \label{eqModifiedHamTermse}\\
    \tilde{A}_q^{(\fg)} &= \prod_{f \in \partial q} \tau_f^{X}, &\quad \tilde{B}_e^{(\fg)} &=  \tau_{O_e}^X \tau_{U_e}^X \sigma_{e}^X \prod_{f \in \delta e}\tau_f^Z.\label{eqModifiedHamTermsm}
\end{align}
where $\delta e$ ($\delta v) $ consists of all faces (edges) incident to the edge $e$ (vertex $v$), and $\partial f$ ($\partial q$) consists of all edges (faces) in the boundary of the face $f$ (volume $q$), and $U_f, O_f$, and $U_e, O_e$ are edges and faces, respectively, depicted in Fig.~\ref{fig3FModifiedWWModel}.

 \begin{figure}[t]%
	\centering
	\includegraphics[width=0.95\linewidth]{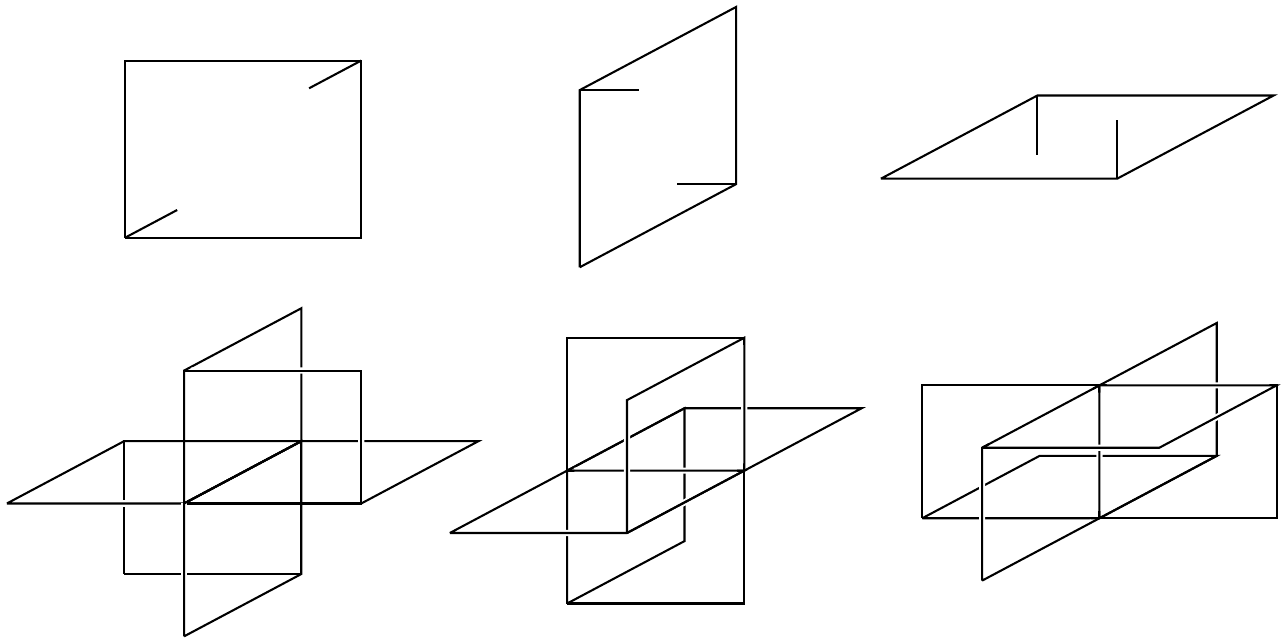} 
	\caption{The \textbf{3F} Walker--Wang plaquette terms after translation of each of the $\tau$ qubits in the original lattice by $\frac{1}{2}(1,1,1)$. $\sigma$ qubits live on edges, while $\tau$ qubits live on faces. The support of the terms $\tilde{B}_f^{(\fr)}$ and $\tilde{B}_e^{(\fg)}$ are shown on top and bottom, respectively. For a given face $f$, the edges $U_f, O_f$ are precisely those depicted that are not in the boundary of the face. Similarly, for a given edge $e$, the faces $U_e, O_e$ are those depicted that are not in the coboundary of the edge. The 1-form constraint terms $\tilde{A}_v^{(\fr)}$ and $\tilde{A}_q^{(\fg)}$ are given by a product of Pauli-$X$ operators on the star of a vertex and the boundary of a cube, respectively.}
	\label{fig3FModifiedWWModel}
 \end{figure}
 
On this lattice, the symmetry $S(\grouponearg)$ can be entirely implemented by a lattice transformation:
 \begin{equation}\label{eq3FTranslationSymmetry}
     S(\grouponearg) = T(w), \quad w = \frac{1}{2}(\pm1, \pm1, \pm1),
 \end{equation}
where it is understood that the $\pm$ sign for each direction can be chosen independently. The symmetry induces the correct permutation action on Hamiltonian terms: namely, $\tilde{B}_f^{(\fr)}$ and $\tilde{B}_e^{(\fg)}$ plaquettes are permuted, as are the 1-form generators $\tilde{A}_v^{(\fr)}$ and $\tilde{A}_q^{(\fg)}$. 
We remark that there are other choices of translation vector that realise the symmetry. One can directly generate the lattice representations for the other 2-cycle symmetries by composing $S(\grouponearg)$ and $S(\grouptwoarg)$.

\subsection{Construction of symmetry defects in stabilizer models for locality-preserving symmetries}\label{subsecWWSymmetryDefects}
Here we present a general construction for implementing symmetry defects in 3D stabilizer models, whenever the symmetry is given by a constant depth circuit with a potential (partial) translation.  
The prescription leverages similar constructions of symmetry defects in 2D systems \cite{bombin2010topologicaltwist, barkeshli2019symmetry, cheng2016translational}.
The construction admits a direct generalisation to a wider class of locality preserving symmetries, 
such as those realized by quantum cellular automata~\cite{haah2018nontrivial}, and we expect that it extends to more general topological commuting projector models.
In particular, we give a prescription for implementing symmetry defects for the $S_3$ symmetries of the \textbf{3F} Walker--Wang model, with explicit examples provided in App.~\ref{secWWSymmetryDefects}.

\subsubsection{Codimension-1 domain walls}
Let us begin by implementing $g$-domain walls in a stabilizer Hamiltonian $H$ with a symmetry $S(g)$ represented by a locality preserving unitary. Consider for simplicity an infinite lattice, that is partitioned by a two-dimensional surface $D$ into two connected halves $L\cup R$ (for example, $D$ may be a lattice plane). 
Our goal is to create a codimension-1 domain wall supported near $D$. 
We decompose the Hilbert space as $\mathcal{H} = \mathcal{H}_{L} \otimes \mathcal{H}_{R}$, where $\mathcal{H}_{L}$ and $\mathcal{H}_{R}$ are the Hilbert spaces for the two halves, which we refer to as the left and right spaces. We require that the partition is such that there is a natural restriction to one of the half spaces, which without loss of generality we assume to be $R$. In particular, we require the restriction $S_{R}(g) = S(g)|_{R}$ of $S(g)$ to $\mathcal{H}_{R}$ to be a well defined map
\begin{equation}
    S_{R}(g): \mathcal{H}_R \rightarrow \mathcal{H}_{R}.
\end{equation}
For any constant depth unitary circuit, there exists a well defined restriction that is unique up to a local unitary that acts within a small neighbourhood of $D$. 
In the presence of translation symmetries we additionally require that the translation is injective on $\mathcal{H}_R$ -- that is, we require that the translation maps one half of the partition to itself. 
Such transformations can be achieved for example if $D$ is a plane or multiple half-planes that meet. 
For $D$ a lattice plane this accommodates half space translations orthogonal to $D$ that are injective but not surjective on $\mathcal{H}_R$. 

With this restriction, the Hamiltonian with a $g$-domain wall is given by conjugating the Hamiltonian $H$ by the restriction of the symmetry $S_{R}(g)$. 
We remark that the defect Hamiltonian differs from $H$ only near the plane $D$. Namely, the restriction $S_{R}(g)$ preserves all Hamiltonian terms that are supported entirely within $\mathcal{H}_{R}$ (as it is a symmetry of $H$), has no affect on the terms that are supported entirely within $\mathcal{H}_{L}$, but may have some nontrivial action on terms supported on both $\mathcal{H}_L$ and $\mathcal{H}_R$. The modified terms supported in the neighbourhood of $D$ realise the $g$-domain wall. Such modified terms commute with each other and the remainder of the Hamiltonian, since their (anti)commutation relations upon restriction to either side of $D$ are preserved by $S_R(g)$.

We remark that when $S(g)$ is a locality preserving, but not onsite, unitary the Hilbert space near the domain wall may be modified. In particular, for a symmetry involving translation, the new Hilbert space may be a strict subset of the old Hilbert space. That is, a subset of qubits in  $\mathcal{H}_R$  near the domain wall $D$ may be ``deleted" (for example if the translational symmetry is not parallel to the $D$~plane).

\subsubsection{Codimension-2 twist defects}\label{subsecTwistPrescription}
We now consider domain walls that terminate in codimension-2 twists. Consider a domain wall $D$ that has been terminated to create a boundary $\partial D$ which we assume is a straight line (in this way, $D$ no longer parititions the lattice into two halves). Let the Hamiltonian be written $H=\sum_{x\in I}h_x$ for some index set $I$, and $d=\max_{x\in I}\{\text{diam}(h_x)\}$ be the max diameter of any term, where $\text{diam}(h_x)$ is the diameter of the smallest ball containing the support of $h_x$ in the natural lattice metric. 
Along the domain wall $D$ we can modify the Hilbert space and Hamiltonian terms following the previous prescription, provided they commute with the bulk Hamiltonian. 
This works away from the boundary of the domain wall $\partial D$. 
Specifically, one can replace all terms $h_x$ with support intersecting $D$ by $S_R(g) h_x S^{-1}_R(g)$, where again $S_R(g)$ is the restriction to one side of the domain wall, which is locally well defined away from $\partial D$. 

In general, this procedure will break down for terms supported within a distance $d$ of $\partial D$, as the modified terms may no longer commute with the neighbouring bulk Hamiltonian terms and so are not added to the Hamiltonian. In order to ensure that all local degeneracy has been lifted in the neighbourhood of $\partial D$, we must find a maximal set of local Pauli terms that commute with the bulk Hamiltonian and domain wall terms. By stabilizer cleaning~\cite{bravyi2009no} there exist a generating set supported on the qubits within a neighbourhood of radius $d$ of $\partial D$, which we label by $N_d(\partial D)$. By Theorem IV.11. of Ref.~\cite{haah2018nontrivial}, this maximal set of local terms admits a translationally invariant generating set (by assumption we have assumed that $\partial S$ and thus the surrounding Hamiltonian terms have a translational invariance along one dimension). We use such a generating set to define our terms along the twist. 

\begin{figure}[t]%
	\centering
	\includegraphics[width=0.92\linewidth]{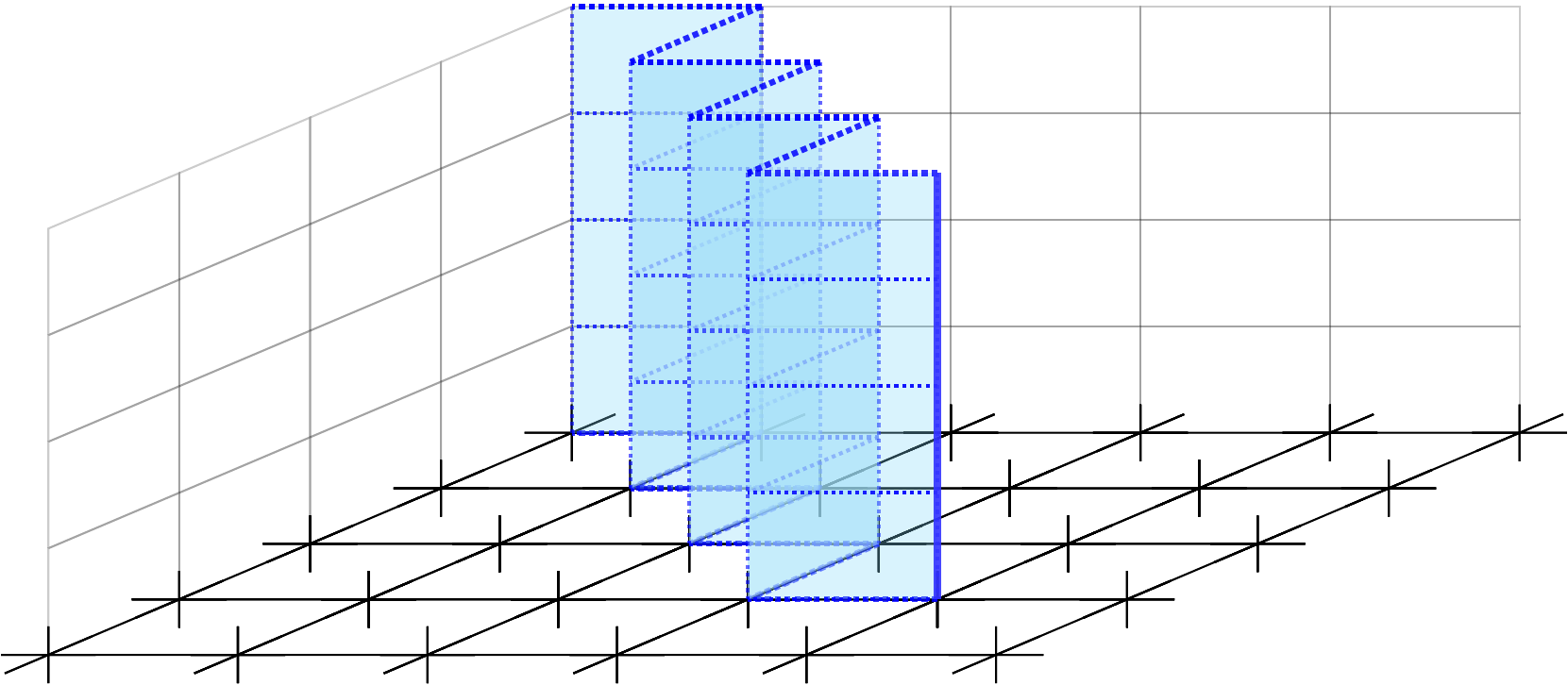}
	\caption{Example of a domain wall plane $D$ for a symmetry $S(\grouponearg)$ ending in a twist (depicted in solid blue) travelling in the $\hat{y}$ direction. The new Hilbert space contains no qubits on any of the shaded edges or faces, leaving a lattice dislocation. }
	\label{fig3FDomainWall}
 \end{figure}
 
\subsubsection{Planes meeting at seams and corners}
For the purposes of discretising domain walls to implement gates from Sec.~\ref{sec3FTQCScheme} on the lattice we are required to consider configurations of two or three domain wall planes that meet at 1D seams and 0D corners, along with twists defect lines that change directions at 0D corners.
If the planes are constructed using different symmetries $S(g)$ or different translations, then Hamiltonian terms in the neighbourhood of seams can be constructed in the same way as the twists (utilising Theorem IV.11. of Ref.~\cite{haah2018nontrivial}). 
Hamiltonian terms in the neighbourhood of a corner where a twist changes direction or where distinct domain wall planes meet can be again computed by finding a maximal set of mutually commuting terms that commute with the surrounding Hamiltonian, which is a finite constant sized problem (and thus can be found by exhaustive search), as such these features are contained within a ball of finite radius.

\begin{figure*}[t]
    \center
	\includegraphics[width=0.48\linewidth]{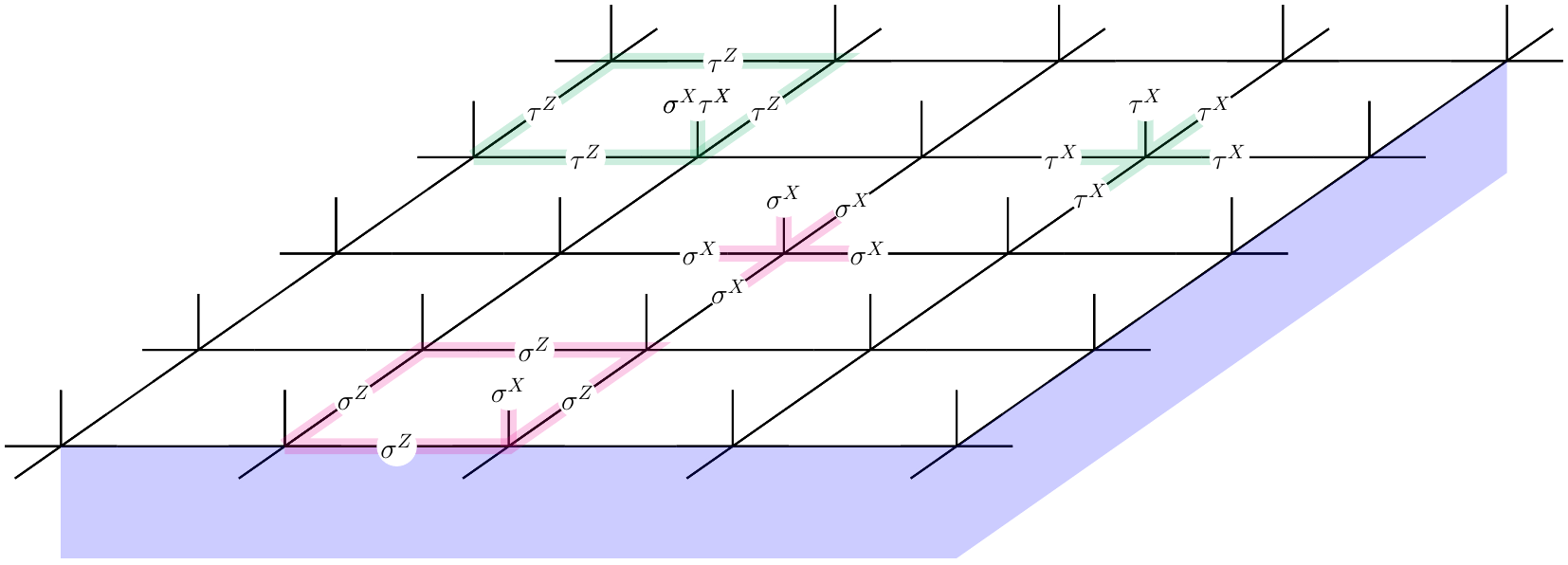} 
	\includegraphics[width=0.48\linewidth]{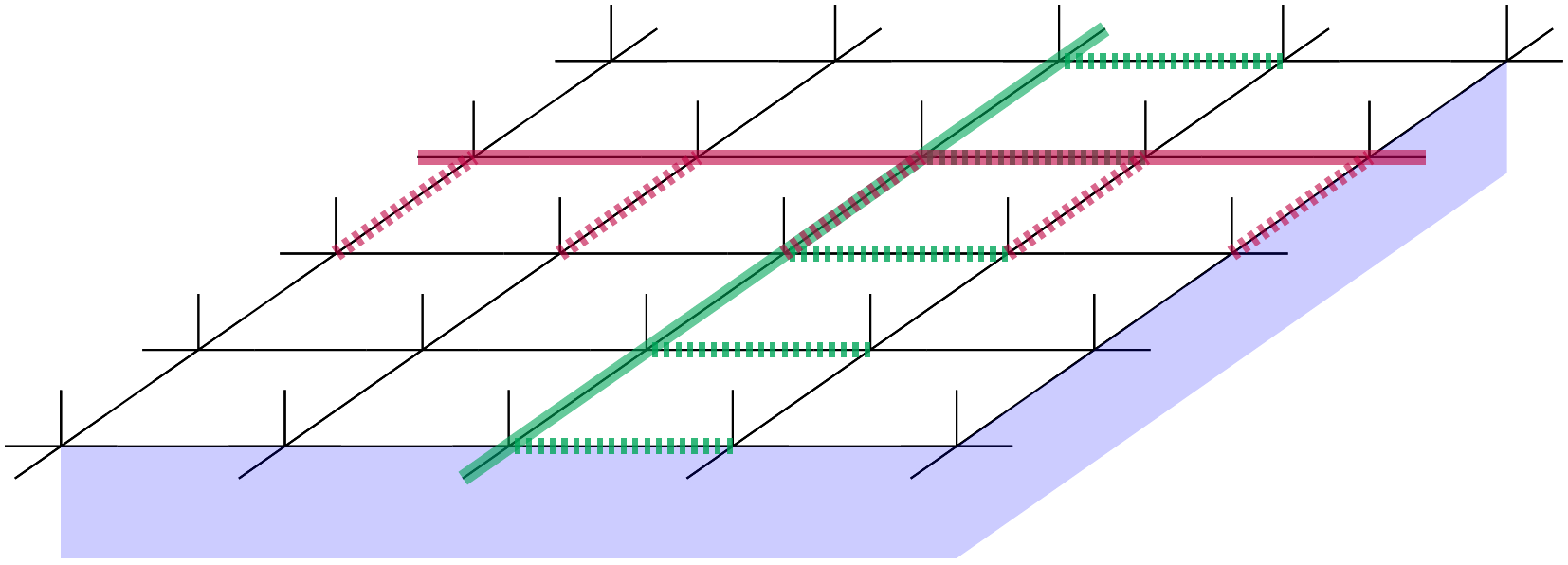} 
	\caption{Boundary of $H_{\textbf{3F}}$. The blue shaded region depicts vacuum (i.e., region with no qubits), and the bulk of $H_{\textbf{3F}}$ lies above the plane. On the left we depict the stabilizers on the boundary: Truncated versions of $A_v^{(\fr)}$ and $B_f^{(\fr)}$ shaded in red, and $A_v^{(\fg)}$ and $B_f^{(\fg)}$ shaded in green. On the right we depict the support of logical operators of Eqs.~(\ref{eq3FLogicalOperators1}), (\ref{eq3FLogicalOperators2}) -- two cycles $c$, $c'$ are depicted by solid red and green lines, while the links belonging to $c_O$, $c'_O$ are depicted by dashed lines. For example, if we take periodic boundary conditions (such that the boundary is a torus), then the two operators $l^{\fr}_c$ and $l^{\fg}_{c'}$ form anti-commuting pairs of logical operators.}
	\label{fig3FBoundary}
\end{figure*}

\subsection{Boundaries of $H_{\textbf{3F}}$}\label{secWWTimeLikeBoundaries}
Finally, we review boundaries of the \textbf{3F} Walker--Wang model. 
On a manifold with boundary, the Walker--Wang model admits a canonical smooth boundary condition~\cite{walker20123+} that supports a topological phase described by the input anyon theory -- in this case the \textbf{3F} anyon theory, as described in Ref.~\cite{burnell2013exactly}.

To be more precise, one may terminate the lattice with smooth boundary conditions as depicted in Fig.~\ref{fig3FBoundary}. The Hamiltonian terms for the boundary can be obtained by truncating the usual bulk terms, see Fig.~\ref{fig3FBoundary}. The boundary supports a topology dependent ground space degeneracy of $2^{2g}$ for an orientable, connected boundary with genus $g$. We can view the ground-space of the boundary as a code with certain logical operators that form anti-commuting pairs. The logical operators come in two types. Let $c$ be a closed cycle on the boundary, then let
\begin{align}
    l^{\fr}_c &= \prod_{i \in c} \sigma^Z_i \prod_{j \in c_O} \sigma_j^X \label{eq3FLogicalOperators1},\\
    l^{\fg}_c &= \prod_{i \in c} \tau^Z_i \prod_{j \in c_O} \sigma_j^X \tau_j^X \label{eq3FLogicalOperators2},
\end{align}
where $c_O$ is a set of links ``over'' the cycle $c$, depicted by dashed lines in Fig.~\ref{fig3FBoundary}. Two operators $l^{\fr}_c$ and $l^{\fg}_{'c}$ anticommute if and only if $c$ and $c'$ intersect an odd number of times and two operators of the same type commute. Representative logical operators can be found by choosing nontrivial cycles $c$ of the boundary.

As described in ~\cite{burnell2013exactly}, the \textbf{3F} anyons are supported as excitations on the boundary. Such excitations correspond to flipped boundary plaquettes $B_f^{(\fr)}$ and $B_f^{(\fg)}$, and can be created at the end of string operators obtained as truncated versions of the loop operators of Eqs.~(\ref{eq3FLogicalOperators1}),~(\ref{eq3FLogicalOperators2}). Further, symmetry defects from the bulk that intersect the boundary give rise to defects on the 2D boundary, behaving as described in Sec.~\ref{sec3FTQCScheme}. Thus the boundary of the \textbf{3F} Walker--Wang model faithfully realises the \textbf{3F} anyon theory and its symmetry defects. In the following section we show how to perform fault-tolerant MBQC with these states.

\section{Fault-tolerant measurement-based quantum computation with Walker--Wang resource states}\label{sec3FMBQC}
Measurement-based quantum computation provides an attractive route to implement the topological computation scheme introduced in Sec.~\ref{sec3FTQCScheme}. 
The computation proceeds by implementing single spin measurements on a suitably prepared resource state -- in this case the ground state(s) of the Walker--Wang model introduced in the previous section. 
In this section we introduce the general concepts required to implement fault-tolerant MBQC with Walker--Wang resource states, focusing on the \textbf{3F} anyon theory example. We begin by reviewing the cluster-state scheme of Ref.~\cite{Rau06} and then recast it in the Walker--Wang framework, before presenting our construction for \textbf{3F} MBQC. We emphasize that the architectural and resource requirements for this \textbf{3F} MBQC scheme are very similar to that of the fault-tolerant MBQC scheme of Ref.~\cite{Rau06}. In particular, the resource states can be prepared with a Clifford circuit acting on qubits arranged on a 2D grid (with only nearest-neighbour interactions).

\subsection{Warm-up: topological cluster state MBQC in the Walker--Wang framework}\label{subsecComparisonToRauss}
The simplest and most well known example of fault-tolerant MBQC is the topological cluster state model from Ref.~\cite{Rau06}.
As a warm-up for what's to come, we explain how this model can be understood as a Walker--Wang model based on the toric code anyon theory.

Up to some lattice simplifications (which we show below), the topological cluster state model~\cite{Rau06} is is prescribed by the Walker--Wang construction using the toric code anyon theory $\mathcal{C}_{\textbf{TC}} = \{1,e, m, \epsilon\}$ as the input. The toric code anyon theory emerges as the fundamental excitations of the toric code~\cite{kitaev2003fault}, they have the following $\mathbb{Z}_2\times\mathbb{Z}_2$ fusion rules:
\begin{equation}
    e \times m = \epsilon, \qquad e\times e = m \times m = 1,
\end{equation}
with modular $S$ matrix the same as \textbf{3F} as in Eq.~(\ref{eq3FModularMatrices}), and $T$ matrix given by $T= \text{diag}(1,1,1,-1)$. The Walker--Wang construction can be used with this input to give a Hamiltonian with plaquette terms as per Fig.~\ref{figRaussendorfWWModel}~(top) along with the same vertex terms as Eqs.~(\ref{eqHamTermse}),~(\ref{eqHamTermsm}). To obtain the more familiar stabilizers of the three-dimensional topological cluster state of Ref.~\cite{Rau06} -- depicted in Fig.~\ref{figRaussendorfWWModel}~(bottom) -- we simply translate all $\tau$ qubits by $\frac{1}{2}(1,1,1)$, as in Eq.~(\ref{eqLatticeTransformation}).

The Walker--Wang construction provides a useful insight into topological quantum computation with the 3D cluster state. In particular, the (unique on any closed manifold) ground state of the toric code Walker--Wang model consists of a superposition over closed anyon diagrams. We interpret the basis states $\ket{++}, \ket{-+}, \ket{+-}, \ket{--}$ on each link as hosting $1,e,m, \epsilon$ anyons, respectively. The ground state is then
\begin{equation}
     \ket{\psi_{\textbf{TC}}} = \sum_{c \in \mathcal{D}} \phi(c) \ket{c}, \quad \phi(c) = (-1)^{\text{linking}(c)},
\end{equation}
where $\mathcal{D}$ is the set of all basis states corresponding to closed anyon diagrams with valid fusion rules that can be created via local moves, and $\text{linking}(c)$ is the linking number of the $e$ and $m$ anyon worldlines.

The computation on this state proceeds by measuring all qubits in the local anyon basis (i.e., $\ket{++}, \ket{-+}, \ket{+-}, \ket{--}$),  projecting it into a definite anyon diagram which we call a history state. As each measurement outcome is in general random, the history state produced is also random. This leads to a outcome-dependent Pauli operator that needs to be applied (or kept track of) to ensure deterministic computation. This Pauli operator is inferred from measurement outcomes of operators known as \textit{correlation surfaces} for each gate \cite{raussendorf2003measurement, Rau06}, which measure the anyon flux between different regions. The computation is fault-tolerant because of the presence of the $\zz_2\times \zz_2$ 1-form symmetry: errors manifest as violations of anyon conservation in the history state and can be accounted for and corrected. 

To implement logical gates, one can use a combination of boundaries and symmetry defects to encode and drive computation. The anyon theory enjoys a $\zz_2$ symmetry: $e \leftrightarrow m$ (which on the usual cluster-state lattice with qubits on faces and edges can be realized by the same translation operator as Eq.~(\ref{eq3FTranslationSymmetry})). Twists defects corresponding to this $\zz_2$ symmetry can be implemented in this lattice using the prescription of Sec.~\ref{subsecWWSymmetryDefects} and can be braided and fused to implement logical gates. (Another method for constructing defects is given by Ref.~\cite{brown2020universal} -- although is distinct from the method proposed in Sec.~\ref{subsecWWSymmetryDefects}.). We remark that braiding these defects is not Clifford-complete. To make the scheme Clifford complete, one can introduce boundaries, of which there are two types (each boundary can condense either $e$ or $m$ anyons)~\cite{Rau06}.

In what follows, we describe the topological MBQC scheme based on the \textbf{3F} theory. 

\begin{figure}[t]%
	\centering
	\includegraphics[width=0.95\linewidth]{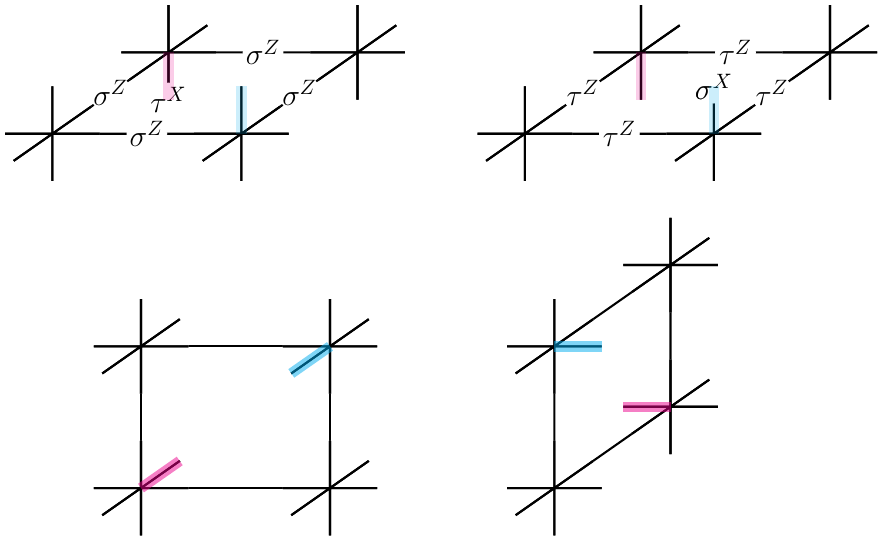} 
    \includegraphics[width=0.95\linewidth]{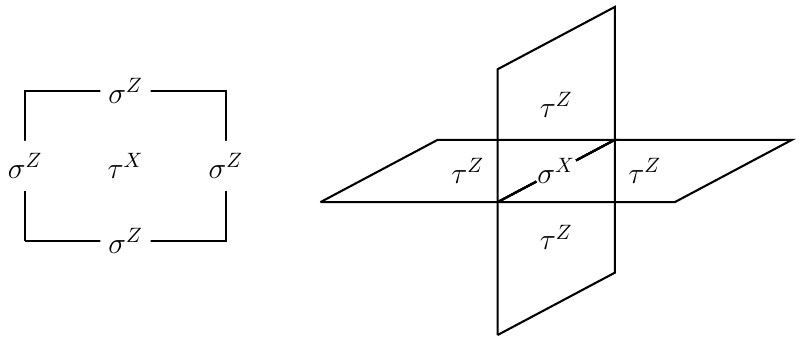} 
	\caption{(top) The Walker--Wang construction applied to the toric code anyon theory $\mathcal{C}_{\textbf{TC}}$ gives the plaquette terms depicted above. Terms on different plaquettes can be obtained by translating and rotating according to the correct orientation, as depicted by the blue and red legs. (bottom) The 3D cluster state terms obtained after all $\tau$ qubits have been translated by $\frac{1}{2}(1,1,1)$. All terms are rotationally symmetric on this lattice. }
	\label{figRaussendorfWWModel}
 \end{figure}

\subsection{3-Fermion topological MBQC}
We now describe how to implement our \textbf{3F} topological quantum computation scheme using an MQBC approach based on Walker--Wang resource states. 
A high level description of the computation scheme is depicted in Fig.~\ref{figWWMBQC}. 

\subsubsection{The \textbf{3F} resource state}
The resource state upon which measurements are performed is given by the ground state of the Walker--Wang Hamiltonian $H_{\textbf{3F}}$ with defects as defined in Sec.~\ref{sec3FWW}, which is a stabilizer model.
This resource state can be understood as a blueprint for the computation and we denote the stabilizer group that defines it by $\mathcal{R} \leq \mathbb{P}_n$ (where $\mathbb{P}_n$ is the Pauli group on $n$ qubits). In particular, $\mathcal{R}$ is generated by all the local terms of $H_{\textbf{3F}}$, and the resource state is a $+1$-eigenstate of all elements of $\mathcal{R}$.

It is instructive to think of one direction of the lattice, say the $\hat{y}$ direction, as being simulated time. For simplicity, we choose the global topology of the lattice to be that of the 3-torus such that the Hamiltonian(s) contain a unique ground state. Of course, one may consider boundaries that support a degenerate ground-space (with \textbf{3F} topological order and possible symmetry defects) as described in Sec.~\ref{secWWTimeLikeBoundaries}, which can be used as the input and output encoded states for quantum computations. However, we remark here that all computations may be performed in the bulk (i.e., with periodic boundary conditions) with all boundaries of interest being introduced by measurement. 

In order to perform computations consisting of a set of preparations, gates and measurements, one prepares the ground state of the \textbf{3F} Walker--Wang Hamiltonian with symmetry defects according to a discretised (on the cubic lattice) version of the topological specification of each gate in Sec.~\ref{sec3FTQCScheme} following the microscopic prescription of Sec.~\ref{sec3FWW}, with gates concatenated in the natural way. 
As the resource state is the ground-state of a stabilizer Hamiltonian, one can fault-tolerantly prepare it using a constant depth Clifford circuit (for example, one may define Clifford gadgets to measure each Hamiltonian term~\cite{steane1997active}, and then compose them together). 

\textbf{Hardware implementations.}
Preparing the resource state can be performed in different ways, depending on the hardware platform and primitives. Despite being a 3D resource state, the computation can be performed on a 2D architecture with only local qubit connectivity. In particular, the entire resource state need not be prepared all at once, and can instead be prepared and measured with only a 2D slab of the resource state being active at any point in time (following for example, Ref.~\cite{Raussendorf07}). Thus the preparation and measurement circuit can be mapped to a Clifford circuit acting on qubits on a 2D layout with local qubit connections, which is possible in many currently pursued architectures, for example photonic qubit architectures, neutral atom architectures, and more~\cite{rudolph2017optimistic,bartolucci2021fusion, krinner2022realizing,acharya2022suppressing, ryan2022implementing, sundaresan2022matching, fukui2018high,sahay2023high}. With access to flying qubits (such as photons), one can prepare certain resource states in a quasi-1D fashion as in Refs.~\cite{wan2021fault, bombin2021interleaving} and it may be interesting to consider this for more general Walker--Wang resource states to further reduce the hardware requirements. It would also be interesting to consider fusion-based versions~\cite{bartolucci2021fusion} of this approach. 

\begin{figure}[t]
	\centering
	\includegraphics[width=0.95\linewidth]{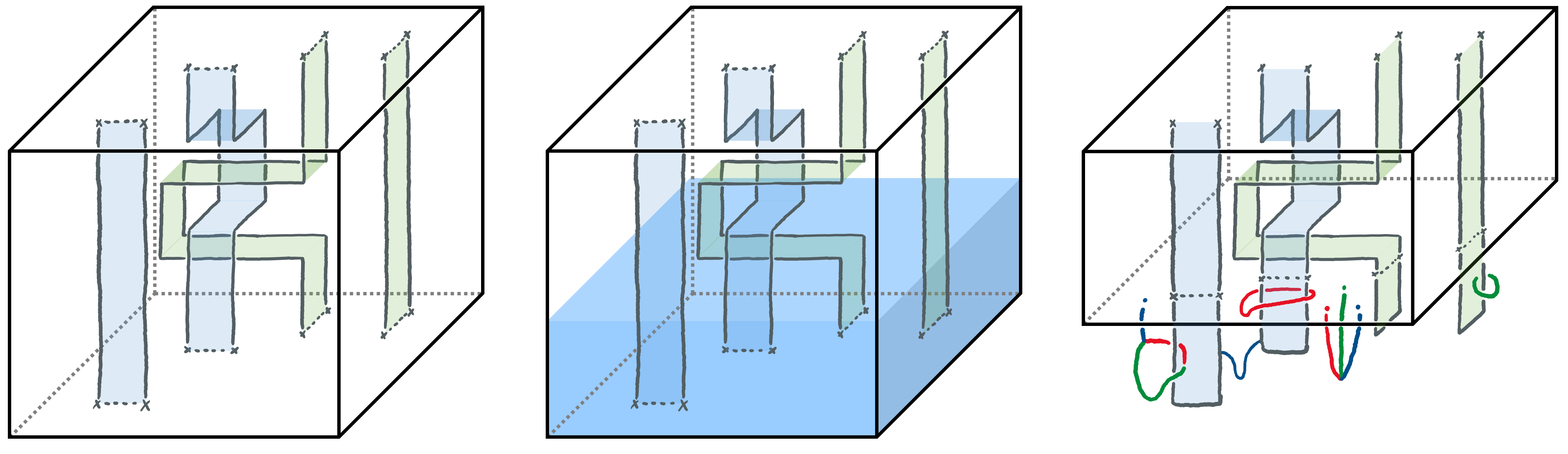} 
	\caption{Fault-tolerant MBQC using the \textbf{3F} Walker--Wang model. (left) Defects and twists can be discretised to live on 2-chains of the lattice and their boundary. (middle) Measurements in the fermion basis in the blue region drives the computation. (right) The post measured state is given by a fixed fermion worldline string-net. Any violations of the $\zz_2\times\zz_2$ conservation at each vertex results from an error, and is detected by the vertex operators whose outcomes are inferred from the local measurements.}
	\label{figWWMBQC}
\end{figure}

\subsubsection{Topological boundary states through measurement}\label{secPreparing3Fstates}
Much like the traditional approaches to fault-tolerant MBQC~\cite{RBH,Rau06,brown2020universal}, we can understand the measurements as propagating and deforming topologically-encoded states (or in another sense as encoded teleportation). 
To understand this more precisely, and develop intuition about how the topological computation proceeds, we begin with an example. 
Consider the  ground-state of the $\textbf{3F}$ Walker--Wang model on the lattice $\mathcal{L}$. We partition the lattice into three segments $\mathcal{L} = A \sqcup C \sqcup B$ as depicted in Fig.~\ref{fig3FMeasuringBoundary}. To begin with, we consider the case where all the sites in $C$ are measured in the fermion basis -- i.e., in $\sigma_i^X$ and $\tau_i^X$ -- and where $A$ and $B$ are unmeasured. 

Firstly, we observe that the post-measured state supports two bulk \textbf{3F} Walker--Wang ground states in $A$ and $B$, with \textbf{3F} boundary states on the interface surfaces $\partial A$ and $\partial B$. The boundary states are precisely those described in Sec.~\ref{secWWTimeLikeBoundaries}, as one can verify that the post-measured state is stabilized by the same truncated stabilizers of Fig.~\ref{fig3FBoundary} up to a sign. Even in the absence of errors, these boundary states will in general host \textbf{3F} anyons as excitations which live at the end of strings of $-1$ measurement outcomes of $\sigma_i^X$ and $\tau_i^X$. 

These boundary states are maximally entangled. To show this, we introduce the concept of a \textit{correlation surface}, which are certain stabilizers of the resource state, that agree with the measurements in $C$ and restrict to logical operators on the boundaries $\partial A$ and $\partial B$. Namely, we define two planes in the $xy$ and $zy$ directions, $c^{(xy)}$ and $c^{(zy)}$, as per Fig.~\ref{fig3FMeasuringBoundary}, and define the operators 
\begin{align}
    S_{\text{r}}(c^{(xy)}) = \prod_{i \in \partial c^{(xy)}} \sigma_i^Z \prod_{j \in  c_O^{(xy)}}\sigma_j^X  \prod_{k \in  c_U^{(xy)}}\sigma_k^X \tau_k^X, \\
    S_{\text{g}}(c^{(zy)}) = \prod_{i \in \partial c^{(zy)}} \tau_i^Z \prod_{j \in  c_O^{(zy)}}\sigma_j^X \tau_j^X  \prod_{k \in  c_U^{(zy)}} \tau_k^X,
\end{align}
where $c_O^{(xy)}$ ($c_O^{(zy)}$) and $c_U^{(xy)}$ ($ c_U^{(zy)}$) denote the sets of edges perpendicular to and on each side of $c^{(xy)}$ ($c^{(zy)}$). Namely, $c_O^{(xy)}$ ($c_O^{(zy)}$) is the set of edges over the surface $c^{(xy)}$ ($c^{(zy)}$), i.e., on the same side as the dashed edges in Fig.~\ref{fig3FMeasuringBoundary}, while $c_O^{(xy)}$ ($c_O^{(zy)}$) is the set of edges under the surface $c^{(xy)}$ ($c^{(zy)}$), i.e., on the opposite side of the dashed edges in Fig.~\ref{fig3FMeasuringBoundary}.

The operators $S_{\text{r}}(c^{(xy)})$ $S_{\text{g}}(c^{(zy)})$ are stabilizers for the Walker--Wang ground state and we refer to them as \textit{correlation surfaces}: they are products of plaquette terms $B_f^{(\fr)}$ and $B_f^{(\fg)}$ in the $c^{(xy)}$ and $c^{(zy)}$ planes, respectively. They can be viewed as world-sheets of the \textbf{3F} boundary state logical operators (they are the analogues of the correlation surfaces in topological cluster state computation of Ref.~\cite{RBH}). In particular, they restrict to logical operators of the \textbf{3F} boundary states on $\partial A$ and $\partial B$ and can be used to infer the correlations between the post-measured boundaries. Namely, we have that the post-measured state is a +1-eigenstate of 
\begin{align}
    \pm l^{\fr}_{\partial c^{(xy)}\cap A} \otimes l^{\fr}_{\partial c^{(xy)}\cap B},\\
    \pm l^{\fg}_{\partial c^{(zy)}\cap A} \otimes l^{\fg}_{\partial c^{(zy)}\cap B}, 
\end{align}
where each factor is a logical operator for the boundary code, as defined in Eqs.~(\ref{eq3FLogicalOperators1}),~(\ref{eq3FLogicalOperators2}) and where the $\pm$ signs are determined by the outcome of the measurements along the correlation surface in $C$. These are the correlations of a maximally entangled pair. Depending on the topology of $\partial A$ and $\partial B$ the boundary state may involve multiple maximally entangled pairs (e.g., 2 pairs if the boundary states are supported on torii). 
We remark that one can construct equivalent, but more natural, correlation surfaces by multiplying with vertex stabilizers $A_v^{(\fr)}$ and $A_v^{(\fg)}$ to obtain the bulk of the correlation surfaces for $S_{\text{r}}(c^{(xy)})$ and $S_{\text{g}}(c^{(zy})$ in terms of a product of $\tau_i^X$ and $\sigma_i^X$ on one side of the surface, respectively.

Importantly, if the region $B$ was prepared in some definite state and measured, the logical state would be teleported to the qubits encoded on the surface at $\partial A$. Conceptually, at any intermediate time during the computation, we may regard the state as being encoded in topological degrees of freedom on a boundary normal to the direction of information flow. This picture holds more generally, when the information may be encoded in twists, and where the propagation of information is again tracked through correlation surfaces that can be regarded as world-sheets of the logical operators. 

\begin{figure}[t]%
	\centering
	\includegraphics[width=0.85\linewidth]{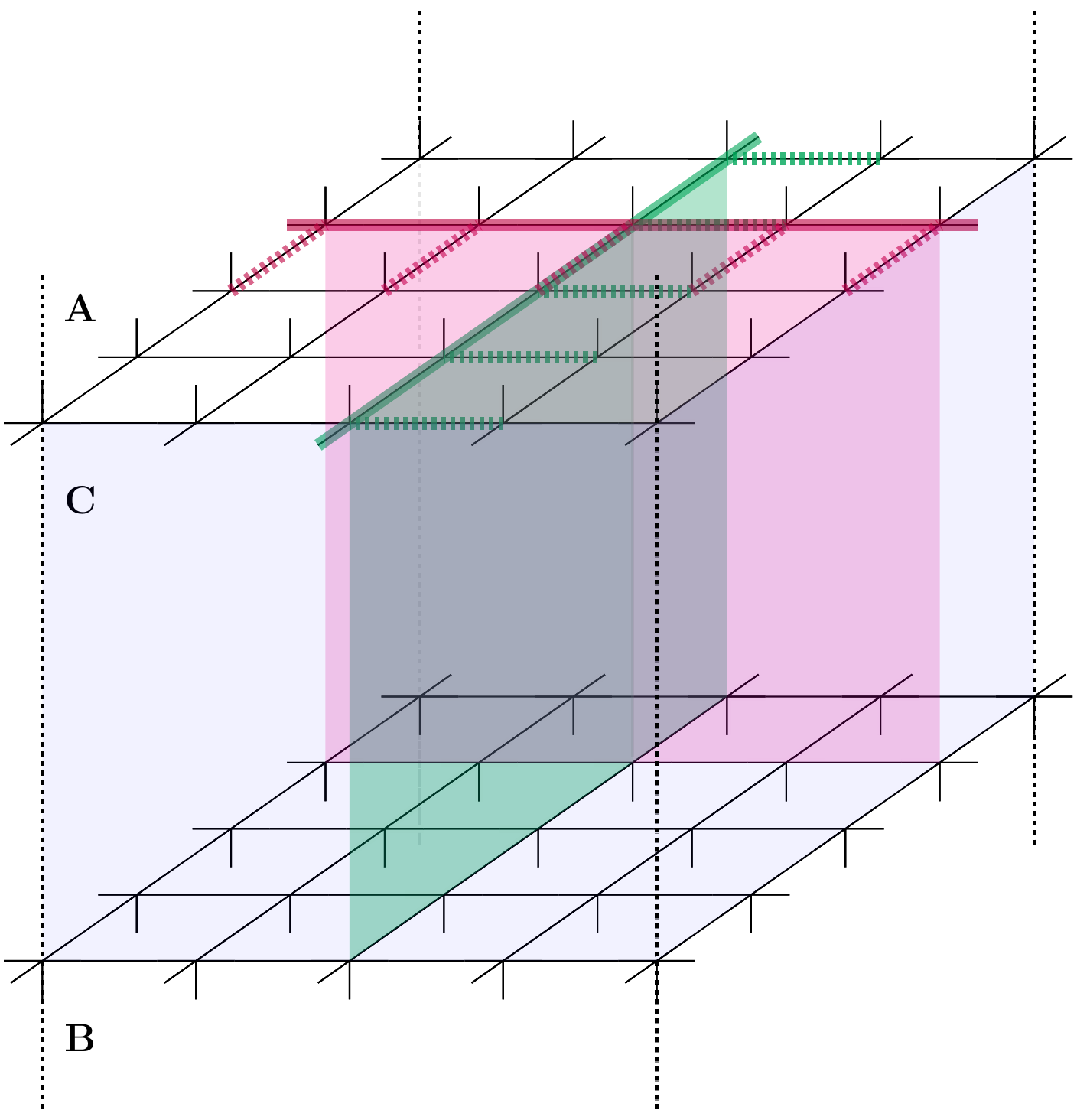} 
	\caption{Preparing \textbf{3F} surface states on the boundaries of $A$ and $B$ by measuring all sites in $C$. The planes $c^{(xy)}$ and $c^{(zy)}$ are depicted in red and green, respectively with the cycles $\partial c^{(xy)}$ and $\partial c^{(zy)}$ on their boundary. The set of links $c_O^{(xy)}$ and $c_O^{(zy)}$ are the set of links perpendicular to the surfaces, on the same side as the dashed lines on $\partial A$. The set of links $c_U^{(xy)}$ and $c_U^{(zy)}$ are on the opposite side.}
	\label{fig3FMeasuringBoundary}
\end{figure}

\subsubsection{Measurement patterns, 1-form symmetries and correlation surfaces.}\label{secCorrelationSUrfaces}

We now consider the general setting for fault-tolerant MBQC with the Walker--Wang resource state. 
The computation is then driven in time by applying single qubit measurements to a resource state describing the Walker--Wang ground state with defects. Such measurements are sequentially applied and the outcomes are processed to determine Pauli corrections, logical measurement outcomes as well as any errors that may have occurred. We label by $\mathcal{M}\subseteq \mathbb{P}_n$ the group generated by the single qubit measurements. For the \textbf{3F} Walker--Wang resource state, we measure in the local fermion basis to project onto a definite fermion wordline occupation state, giving 
\begin{equation}
    \mathcal{M} = \langle \sigma_i^X, \tau_i^X ~|~ i \in \mathcal{L} \rangle.
\end{equation}
We remark for magic state preparation, as per Sec.~\ref{secMagicStatePreparation}, the measurement pattern must in general be modified in a manner that depends on the implementation. Additionally, the measurement pattern may be locally modified in the vicinity of a twist defect, again depending on the implementation. For the twist identified in the previous section, we require a chain of Pauli-$Y$ measurements on the qubits uniquely determined by the 1-form operators of Fig.~\ref{fig3FTwists}. The post-measured state can be regarded as a classical \textit{history state} with definite fermion worldlines.

Individual measurement outcomes are random and in general measurements result in a random fermion worldline occupation on each link of the lattice. However, there are constraints in the absence of errors. Namely, at each bulk vertex the the $\zz_2\times \zz_2$ fermion charge must be conserved. This bulk conservation is measured by the operators $A_v^{(\fr)}, A_v^{(\fg)}$, which belong to both the resource state and measurement group, $A_v^{(\fr)}, A_v^{(\fg)} \in \mathcal{R}\cap \mathcal{M}$. 
The conservation law is modified near defects and domain walls, so too are the corresponding operators from $\mathcal{R}\cap \mathcal{M}$. Therefore, in the absence of errors, due to membership in $\mathcal{R}$, measurement of any operator from $\mathcal{R}\cap \mathcal{M}$ would deterministically return $+1$, signifying the appropriate fermion conservation. Due to membership in $\mathcal{M}$, the outcome of these operators can be inferred during computation as the measurement proceed. 

The vertex operators generate a symmetry group 
\begin{equation}
\mathcal{S} = \langle A_v^{(\fr)}, A_v^{(\fg)} ~|~ v \in \mathcal{L} \rangle,
\end{equation}
where we assume that each vertex operator is suitably modified near symmetry defects and domain walls.
This is known as a $\zz_2 \times \zz_2$ 1-form symmetry group because it consists of operators supported on closed codimension-1 submanifolds of the lattice~\cite{Gaiotto2015}. In terms of the Walker--Wang model for the \textbf{3F} theory, operators in $\mathcal{S}$ measure the fermionic flux through each contractible region of the lattice -- which must be net neutral in the groundstate. 

Even in the absence of errors, the randomness of measurement outcomes can result in fermionic worldlines (in the post-measured state) that nontrivially connect distinct twists. In particular, at each point in the computation, this randomness results in a change in the charge on a twist line and can be mapped to an outcome-dependent logical Pauli operator that has been applied to the logical state. This outcome-dependent Pauli operator is called the logical Pauli frame, and can be deduced by the outcomes of the correlation surfaces (as we have seen in the example of Sec.~\ref{secPreparing3Fstates}).

The correlation surfaces are obtained for each preparation, gate, and measurement. They are stabilizers of the resource state that can be viewed as topologically nontrivial 1-form operators that enclose (and measure) the flux through a region and thus the charge on relevant sets of defects in the history state. We define correlation surfaces for each operation in App.~\ref{appProofOfGates}. Correlation surfaces are not uniquely defined: multiplication by any 1-form operators $s\in \mathcal{S}$ produces another valid correlation surface that is logically equivalent (i.e., will determine the same logical Pauli frame). 
For a given operation, we label the set of all correlation surfaces up to equivalence under $\mathcal{S}$ by $\overline{\mathcal{S}}$. This equivalence allows us to map between different representative logical operators as explained in Sec.~\ref{sec3FEncodings}.

\begin{figure}[t]%
    \centering
    \includegraphics[width=0.6\linewidth]{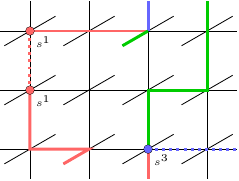}
    \caption{Syndromes observed in Walker--Wang MBQC. Lines are color-coded according to the observed measurement outcomes corresponding to the basis $\ket{\textbf{1}} := \ket{++}$, $\ket{\fr} := \ket{-+}$, $\ket{\fg} := \ket{+-}$, $\ket{\fb} := \ket{--}$. Possible errors producing the observed syndrome are displayed by dashed lines. Nontrivial syndromes $s_v = (a,b) \in \zz_2^2$ on each vertex are observed due to violations of the $\zz_2^2$ charge flux on each vertex and can be inferred from the measurement outcomes of $(A_v^{(\fr)}, A_v^{(\fg))})$. For example, $s^1 = (1,0)$ and $s^3 = (1,1)$ arises from $\fr$ and $\fb$ string errors, as depicted.}
    \label{figLogicalError2}
\end{figure}
 
\subsubsection{Errors, fermion parity, and decoding.}
Errors may occur during resource state preparation, computation, and measurement. For simplicity, we focus on Pauli errors acting on the resource state along with measurement errors. We refer to this hardware agnostic error model allows us to understand the performance of the Walker--Wang MBQC scheme in terms of its fundamental topological properties, ignoring details of how the state is prepared (which depend on the hardware-specific implementation).  

We firstly note that $\sigma_i^Z$, $\tau_i^Z$ and $\sigma_i^Z \tau_i^Z$ errors acting on the resource state result in flipped $\sigma_i^X$ and $\tau_i^X$ measurement outcomes. In the resource state wavefunction, they can be thought of as creating $\fr$, $\fg$, and $\fb$ fermion string segments, respectively. On the other hand, $\sigma_i^X$, $\tau_i^X$ and $\sigma_i^X \tau_i^X$ errors in the bulk are benign; they commute with the measurements and thus do not affect the measurement outcome (as is the case in topological cluster state computation~\cite{Rau06, RBH}). In the Walker--Wang resource state wavefunction $\sigma_i^X$, $\tau_i^X$ errors can be thought of as creating a small contractible $\fr$, $\fg$, or $\fb$ fermion worldine loop, respectively, linking edge $i$~\cite{walker20123+,burnell2013exactly}. Finally, measurement errors (i.e., measurements that report the incorrect outcome) are equivalent to $Z$-type physical errors that occurred on the state before measurement.

In the post-measured state, these errors manifest themselves as modifications to the classical history state. Detectable errors are those that give rise to violations of the $\zz_2\times \zz_2$ fermion conservation rule (that exists away from the twists) and are thus revealed by $-1$ outcomes of the 1-form symmetry operators $s\in \mathcal{S}$. We consider example configurations in Fig.~\ref{figLogicalError2}. Nontrivial errors are those that connect distinct twist worldlines. Such errors result in the incorrect inferred outcome of the correlation surfaces in $\overline{\mathcal{S}}$, and therefore an incorrect inference of the logical Pauli frame -- in other words: a logical Pauli error. Such a process is depicted in Fig.~\ref{figLogicalError}. If errors arise by local processes then they can be reliably identified and accounted for if twist worldlines remain well separated.

It is possible to correct for violations of the $A_v^{(\fr)}$ and $A_v^{(\fg)}$ sectors independently (although depending on the noise model, it may be advantageous to correct them jointly). In particular if we represent the outcome of all vertex operators $A_v^{(\fr)}, A_v^{(\fg)}$ by two binary vectors $v_{\mathcal{S}}^{(\fr)} \in \zz_2^{|V|}$, $v_{\mathcal{S}}^{(\fg)} \in \zz_2^{|V|}$, where $|V|$ is the number of vertices in the lattice. Then one can apply the standard minimum weight perfect matching algorithm that is commonly used for topological error correction \cite{dennis2002topological,Rau06}. 
The algorithm returns a matching of vertices for each sector, $\fr$ and $\fg$, which can be used to deduce a path of measurement outcomes that need to be flipped to restore local fermion parity (i.e., ensure $s\in \mathcal{S}$ has a $+1$ outcome). 

\begin{figure}[t]%
	\centering
	\includegraphics[width=0.85\linewidth]{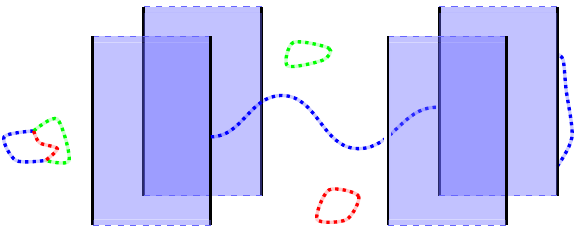} 
	\caption{Undetectable errors in Walker--Wang MBQC depicted by dashed lines. The homologically trivial loops do not result in a logical error. The central error depicted in blue that extends between different twists results in a logical error.}
	\label{figLogicalError}
 \end{figure}

\subsubsection{Threshold performance}

Assuming a phenomenological error model of perfect state preparation, memory and only noisy measurements with rate $p$, the bulk \textbf{3F} Walker--Wang MBQC scheme has a high threshold identical to that of the topological cluster state formulation~\cite{dennis2002topological,RBH} (assuming the same decoder). In particular, under optimal decoding, the scheme has a threshold for noisy measurements of $p=0.033 \pm 0.001$~\cite{RBH,ohno2004phase}. This follows from the fact that the error model and bulk decoding problem is identical to that of topological cluster state computation~\cite{RBH}.  

To obtain more accurate estimates of threshold performance in a realistic setting, one should consider a hardware-motivated error model. For example, for a circuit-level error model preparing the \textbf{3F} Walker--Wang resource state, one may expect a lower threshold than that of the topological cluster-state scheme, owing to the higher-weight stabilizers of the resource state, and each qubit being supported in more stabilizers. However, designing the preparation circuits to limit the spread of errors, and tailoring the decoder based on this circuit may mitigate this, or even lead to threshold improvements. Further, for other platforms such as photonic fusion-based quantum computation, the threshold may even improve. We leave the study of threshold performance under hardware-motivated models to future work.

\subsection{1-form symmetry-protected topological order and Walker--Wang resource states}\label{secSPTorder1form}
 
We remark that while both ground states of the \textbf{3F} and \textbf{TC} Walker--Wang models can be prepared by a quantum cellular automaton, only the \textbf{TC} Walker--Wang model ground state can be prepared from a constant depth circuit~\cite{haah2018nontrivial}. Indeed, the two phases, belong to distinct nontrivial SPT phases under $\zz_2^2$ 1-form symmetries. The topological cluster state model has been demonstrated to maintain its nontrivial SPT order at nonzero temperature~\cite{roberts2017symmetry,roberts2020symmetry}, as has the \textbf{3F} Walker--Wang models\cite{stahl2021symmetry}. By the same arguments as in Refs.~\cite{roberts2017symmetry,roberts2020symmetry}, the \textbf{3F} Walker--Wang model belongs to a nontrivial SPT phase under 1-form symmetries, distinct from the topological cluster state model. 

More generally, the bulk of any Walker--Wang state arising from a modular anyon theory should be SPT ordered under a 1-form symmetry (or appropriate generalisation thereof). One can diagnose the nontrivial SPT order under 1-form symmetries by looking at the anomalous action of the symmetry on the boundary. This anomalous 1-form symmetry boundary action corresponds to the string operators of a modular anyon theory. A gapped phase supporting that anyon theory can be used to realize a gapped boundary condition that fulfils the required anomaly matching condition. This boundary theory can form a thermally stable, self-correcting quantum memory when protected by the 1-form symmetries~\cite{roberts2020symmetry,stahl2021symmetry}.

Thus the Walker--Wang paradigm provides a useful lens to search for (thermally stable) SPT ordered resource states for MBQC. 
However, determining whether these computational schemes are stable to perturbations of the Walker--Wang parent Hamiltonian for the resource state remains an interesting open problem. 
For 1-form symmetry respecting perturbations, at least, we expect the usefulness of the resource state to persist, as the key relation between the 1-form symmetry and (possibly fattened) boundary string operators remains. 
This potentially has important implications for the existence of fault-tolerant, computationally universal phases of matter~\cite{DBcompPhases,Miy10, else2012symmetry,else2012symmetryPRL,NWMBQC,miller2016hierarchy,roberts2017symmetry,bartlett2017robust,wei2017universal,raussendorf2019computationally,roberts2019symmetry,devakul2018universal,Stephen2018computationally,Daniel2019,daniel2020quantum}.

\section{ Lattice defects in a \textbf{3F} topological subsystem code}\label{sec3FSubsystemCode}

\begin{figure}[t]
\center
\includegraphics[scale=0.575]{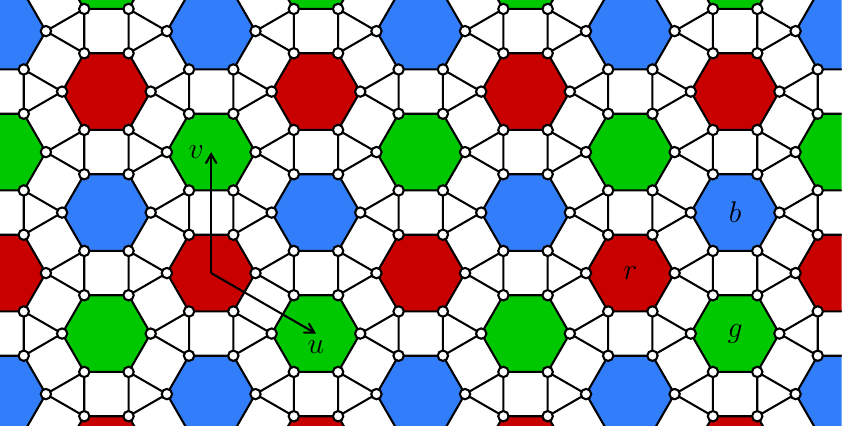}
 \caption{The tricoloring of hexagonal plaquettes used to define the generators of the anomalous $\mathbb{Z}_2 \times \mathbb{Z}_2$ 1-form symmetry. }
 \label{fig:tricolored}
\end{figure}

In Refs.~\cite{bombin2009interacting,bombin2010topologicalsubsystem} a 2D topological subsystem code~\cite{Poulin2005,Bacon2005a} was introduced that supports a stabilizer group corresponding to a lattice realization of the string operators for the \textbf{3F} anyon theory. 
As the gauge generators do not commute, they can be used to define a translation invariant Hamiltonian with tunable parameters that supports distinct phases, and phase transitions between them. 
The model is defined on an inflated honeycomb lattice, where every vertex is blown-up into a triangle, with links labelled by $x,y,z,$ in a translation invariant fashion according to Figs.~\ref{fig:tricolored} \&~\ref{fig:InflatedHexagon}. 
This is reminiscent of Kitaev's honeycomb model~\cite{kitaev2006anyons}, which can also be thought of as a 2D topological subsystem code (that encodes no qubits) with a stabilizer group corresponding to the string operators of an emergent  $\mathbb{Z}_2$ fermion. In this section we review this construction, and show how to implement symmetry defects in the model.

\begin{figure}[t]
\center
{\includegraphics[scale=.85]{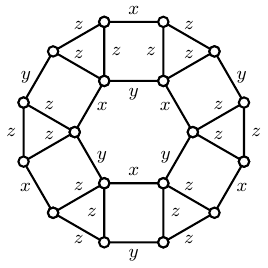}}
    \hspace{.5cm}    
{\includegraphics[scale=.85]{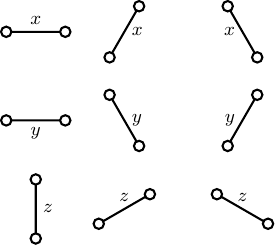}}
 \caption{ (Left) An inflated hexagon.
 (Right) There are three different types of $x$, $y$, and $z$ links in the lattice, respectively. }
 \label{fig:InflatedHexagon}
 \label{fig:Edges}
\end{figure}

The 2D topological subsystem code of Refs.~\cite{bombin2009interacting,bombin2010topologicalsubsystem} is defined on the lattice of Fig.~\ref{fig:tricolored}, with one qubit per vertex. There is one gauge generator per edge, given by 
\begin{align}
K_{\braket{ij}} = \begin{cases}
X_i X_j & \text{if $\langle i j\rangle$ is an $x$-link,} 
\\
Y_i Y_j & \text{if $\langle i j\rangle$ is a $y$-link,} 
\\
Z_i Z_j & \text{if $\langle i j\rangle$ is a $z$-link,} 
\end{cases}
\end{align}
see Fig.~\ref{fig:InflatedHexagon}. 
The Hamiltonian can be written in terms of the gauge generators 
\begin{align}
H = - J_x \sum_{x\text{-links}} K_{\braket{ij}}
- J_y \sum_{y\text{-links}}K_{\braket{ij}}
-J_z \sum_{z\text{-links}}K_{\braket{ij}}
\, ,
\end{align}
where $J_x,J_y,J_z,$ are tunable coupling strengths. 

The group of stabilizer operators that commute with all the gauge generators, and are themselves products of gauge generators, are generated by a $\mathbb{Z}_2 \times \mathbb{Z}_2$ algebra on each inflated plaquette. 
The plaquette algebra is generated by $W_p^X$, $W_p^Z$, and $W_p^Y$ on each plaquette, which satisfy 
\begin{align}
(W_p^{X})^2=(W_p^{Z})^2=\openone \, , 
&&
W_p^X W_p^Z = W_p^Z W_p^X = W_p^Y \, , 
\end{align}
see Fig.~\ref{fig:Generators}.

\subsection{String operators} 

The above plaquette operators are in fact loops of a $\mathbb{Z}_2 \times \mathbb{Z}_2$ algebra of string operators on the boundary of the plaquette. 
To define larger loops of the string operators we make use of a tricoloring of the hexagon plaquettes shown in Fig.~\ref{fig:tricolored}. 
On the boundary of a region $\mathcal{R}$, given by a union of inflated plaquettes on the inflated honeycomb lattice, we have the following $\mathbb{Z}_2 \times \mathbb{Z}_2$ string operators 
\begin{align}
\label{eq:LoopOps}
W_{\partial\mathcal{R}}^r &= 
\prod_{p_r \in \mathcal{R}} W_{p_r}^{Z}
\prod_{p_g \in \mathcal{R}} W_{p_g}^{X}
\prod_{p_b \in \mathcal{R}} W_{p_b}^{Y} \, ,
\\
W_{\partial\mathcal{R}}^g &= 
\prod_{p_r \in \mathcal{R}} W_{p_r}^{Y}
\prod_{p_g \in \mathcal{R}} W_{p_g}^{Z}
\prod_{p_b \in \mathcal{R}} W_{p_b}^{X} \, ,
\\
W_{\partial\mathcal{R}}^b &= 
\prod_{p_r \in \mathcal{R}} W_{p_r}^{X}
\prod_{p_g \in \mathcal{R}} W_{p_g}^{Y}
\prod_{p_b \in \mathcal{R}} W_{p_b}^{Z} \, ,
\end{align}
where $p_r,p_g,$ and $p_b$ stand for red, green, and blue plaquettes, respectively. 
The string operators satisfy the same algebra as the plaquette operators
\begin{align}
(W_{\partial\mathcal{R}}^{r})^2=(W_{\partial\mathcal{R}}^{b})^2=\openone \, , 
&&
W_{\partial\mathcal{R}}^r W_{\partial\mathcal{R}}^b = W_{\partial\mathcal{R}}^b W_{\partial\mathcal{R}}^r = W_{\partial\mathcal{R}}^g \, .
\end{align}

\begin{figure}[t]
\center
{\includegraphics[scale=.85]{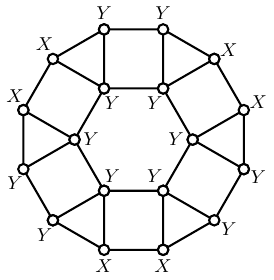}}
    \hspace{.5cm}    
{\includegraphics[scale=.85,trim= 0 -.395cm 0 0]{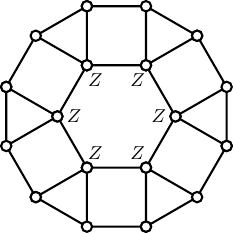}}
 \caption{ (Left) The $W_p^X$ generator on the inflated hexagon.
 (Right) The $W_p^Z$ generator on the inflated hexagon. The $W_p^Y$ generator is given by their product. }
 \label{fig:Generators}
\end{figure}

The loop operators on the boundary of a region $\mathcal{R}$ in Eq.~\eqref{eq:LoopOps} suffice to define red string operators $W^r_{\ell_r}$ on arbitrary open paths $\ell_r$ on inflated edges between red plaquettes, and similarly for green and blue string operators and plaquettes. 
The string operators are given by a product of the elementary string segment operators shown in Fig.~\ref{fig:StringSegments} along the string. 
With the string segment operators shown, the excitations of the $W^r_{\ell_r}$ operator can be thought of as residing on the red plaquettes of the lattice, and similarly for the green and blue plaquettes. 
We denote the superselection sector of the excitation created at one end of an open $W^r_{\ell_r}$ operator by $\fr$, and similarly $\fg$, $\fb$ for green and blue string operators. 
The fusion and braiding processes for these sectors, as defined by the string operators, are described by the \textbf{3F} theory introduced in Sec.~\ref{sec3FPreliminaries}. 

The  set of string operators $W_{\ell_r},W_{\ell_g},W_{\ell_b},$  commute with the Hamiltonian throughout the whole phase diagram 
\begin{align}
[H,W_{\ell_r}^{r}]=[H,W_{\ell_g}^{g}]=[H,W_{\ell_b}^{b}]=0 \, ,
\end{align}
for closed loops $\ell_r,\ell_g,\ell_b$. 
This structure is formalized as an anomalous $\mathbb{Z}_2 \times \mathbb{Z}_2$ 1-form symmetry, with the anomaly capturing the nontrivial $S$ and $T$ matrices of the \textbf{3F} theory associated to the string operators. 
We remark that the Hamiltonian can support phases with larger anyon theories that include the \textbf{3F} theory as a subtheory (due to the factorization of modular tensor categories~\cite{Mueger2002} the total anyon theory is equivalent to a stack of the \textbf{3F} theory with an additional anyon theory). 
In particular, in the $J_z \gg J_x,J_y >0$ limit the Hamiltonian enters the phase of the color code stabilizer model~\cite{bombin2009interacting}. The anyon theory of this model is equivalent to two copies of the \textbf{3F} theory~\cite{bombin2012universal} (or equivalently two copies of the toric code anyon theory~\cite{bombin2012universal,kubica2015unfolding}).

\begin{figure}[t!]
\center
\includegraphics[scale=.8]{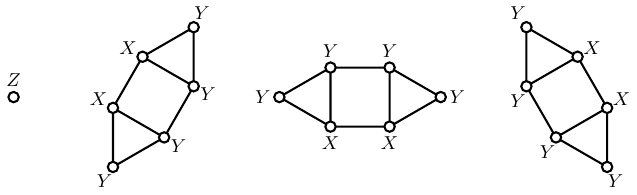}
 \caption{Segments of the string operators that form the anomalous $\mathbb{Z}_2 \times \mathbb{Z}_2$ 1-form symmetry. }
 \label{fig:StringSegments}
\end{figure}

\subsection{Symmetry defects}

The symmetry group of the Hamiltonian is generated by translations $T(u)$ and $T(v)$ along the lattice vectors $u$ and $v$ shown in Fig.~\ref{fig:tricolored}, plaquette centered $\frac{\pi}{3}$-rotations combined with the Clifford operator that implements $X_v \leftrightarrow Y_v$ on all vertices $v$ denoted 
$R_p$, and inflated vertex centered $\frac{2 \pi}{3}$-rotations denoted $R_v$.

\begin{figure}[b]
\center   
\includegraphics[scale=.575]{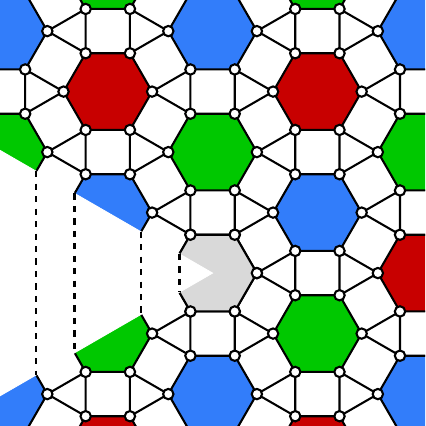}
 \caption{ A $\frac{ \pi}{3}$ lattice disclination on a plaquette that hosts twist defects of a $\mathbb{Z}_2$ symmetry generator.   }
 \label{fig:Z2Defect}
\end{figure}

The \textbf{3F} superselection sectors in the model exhibit \textit{weak symmetry breaking}~\cite{kitaev2006anyons}, or symmetry-enrichment~\cite{barkeshli2019symmetry} under the lattice symmetries giving rise to an $S_3$ action. 

The $\frac{\pi}{3}$ rotation and Clifford operator $R_p$ centered on a red plaquette implements the $(\text{gb})$ symmetry action on the superselection sectors. 
A domain wall attached to a disclination defect with a $\frac{\pi}{3}$ angular deficit can be introduced by cutting a wedge out of the lattice and regluing the dangling edges as shown in Fig.~\ref{fig:Z2Defect}. 
This leads to mixed edges across the cut formed by rejoining broken $x$ and $y$ edges, the Hamiltonian terms on these edges are of the form $X_v Y_v'$ where $v$ is the vertex adjacent to the $x$ portion of the rejoined edge and $v'$ is the vertex adjacent to the $y$ portion of the rejoined edge. 
Assuming the lattice model lies in a gapped phase described by the \textbf{3F} theory, such a lattice symmetry defect supports a non-abelian twist defect $\mathcal{T}_{(\text{gb})}^\pm$, where the $\pm$ is determined by the eigenvalue of the string operator $W_{\ell_r}$ encircling the defect. 
This twist defect is similar to a Majorana zero mode as it has quantum dimension $\sqrt{2}$ and fusion rules given in Sec.~\ref{sec3FPreliminaries}. 
A similar result holds with $\frac{\pi}{3}$ disclination defects centered on green and blue plaquettes hosting $\mathcal{T}_{(\text{rb})}^\pm$ and $\mathcal{T}_{(\text{rg})}^\pm$ twist defects, respectively. 

The $\frac{2 \pi}{3}$ rotation operator $R_v$ centered on an inflated vertex, and also the translations $T(u)$ and $T(v)$, implement the $(\text{rgb})$ symmetry action on the superselection sectors. 
Similar to above a disclination defect with a $\frac{2 \pi}{3}$ angular deficit can be introduced by cutting a wedge out of the lattice and rejoining the dangling edges following Fig.~\ref{fig:Z3Defect1}. 
We can also introduce a dislocation defect as shown in Fig.~\ref{fig:Z3Defect2}. 
Again assuming the lattice model lies in the \textbf{3F} phase, such lattice symmetry defects support non-abelian topological defects with quantum dimension $2$ introduced as $\mathcal{T}_{(\text{rgb})}$ in Sec.~\ref{sec3FDiscussion}.

\begin{figure}[t]
\center  
	{\includegraphics[scale=.575]{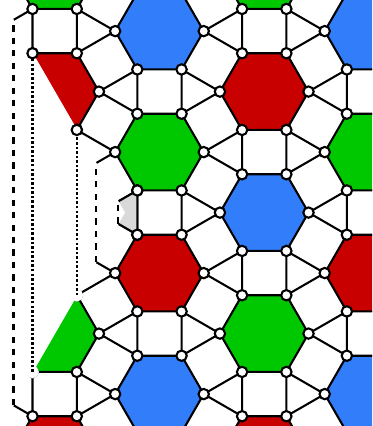}}
   \hspace{.5cm}    
	{\includegraphics[scale=.575]{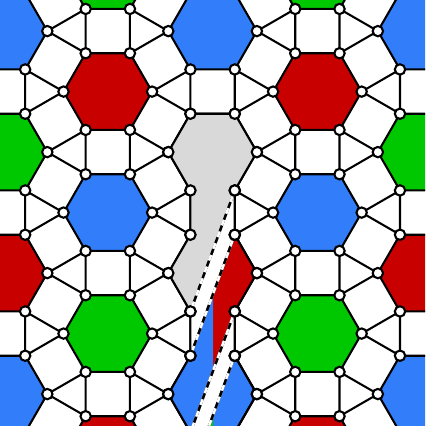}}
 \caption{ (Left) A $\frac{2 \pi}{3}$  lattice disclination on an inflated vertex that hosts twist defects of the $\mathbb{Z}_3$ symmetry generator. 
 (Right) A lattice dislocation on a plaquette that can also host twist defects of the $\mathbb{Z}_3$ symmetry generator.  }
 \label{fig:Z3Defect1}
 \label{fig:Z3Defect2}
\end{figure}

These lattice implementations of the twist defects can in principle be used to realize the defect topological quantum computation schemes introduced in Sec.~\ref{sec3FTQCScheme}. 
This is closely related to the lattice defect based code deformation scheme in Ref.~\cite{Bombin_2011}, where similar defects in the 2D gauge color code were shown to generate all Clifford gates under braiding.
To perform error-correction, we must define the order in which gauge generators are measured to extract a stabilizer syndrome. At each time-step a subset of gauge generators are measured where each of the gauge operators must have disjoint support, for example following Ref.~\cite{bombin2010topologicalsubsystem}. 
In the presence of twists, one must take extra care in defining a globally consistent schedule. A simple (possibly inefficient) approach can be obtained by partitioning the schedule according to the gauge generators along the defect and the bulk separately.

We remark that symmetry defects in the stabilizer color code have been explored previously in Refs.~\cite{yoshida2015topological,teo2014unconventional,kesselring2018boundaries}. 
This presents an alternative route to implement the defect computation scheme of Sec.~\ref{sec3FTQCScheme}. 
This is particularly relevant as the 2D stabilizer color code~\cite{bombin2006topological} is obtained in the limit $J_z \gg J_x,J_y >0$.

\section{Discussion}\label{sec3FDiscussion}

We have presented a general approach to topological quantum computation based on Walker--Wang resource states and their symmetries. As a specific example, we have introduced a universal fault-tolerant quantum computational scheme based on symmetry defects of the \textbf{3F} anyon theory, and how it can be implemented with measurement-based quantum computation on Walker--Wang ground states. Under a phenomenological toy noise model consisting of bit/phase flip errors and measurement errors, the threshold of the \textbf{3F} Walker--Wang computation scheme is equal to that of the well known toric code (or equivalently the topological cluster state scheme) under the same noise model (see e.g.,  Refs.~\cite{dennis2002topological,TopoClusterComp} for threshold estimates). Further investigation under more realistic noise models remains an open problem. 

Our computation scheme based on the defects of the \textbf{3F} anyon theory provides a nontrivial example of the power of the Walker--Wang approach, as the \textbf{3F} anyon theory is chiral and cannot be realized as the emergent anyon theory of a 2D commuting projector model (although it can be embedded into one as a subtheory). 
We hope that this example provides an intriguing step into topological quantum computation using more general anyon schemes and a launch point for the study of further non-stabilizer models. 
In particular, our framework generalizes directly to any abelian anyon theory with symmetry defects, leading to a wide class of potential resource states for fault-tolerant MBQC. 

While we have not tried to optimise the overhead of our gate schemes, the richer defect theory (in comparison with toric code) may lead to more efficient implementations of, for example, magic state distillation. 
In addition to this, leveraging the full $G$-crossed theory of the \textbf{3F} anyon theory (studied in \cite{teo2015theory}) could lead to further improvements and more efficient logic gates, arising from the possibility of additional fusions and braiding processes. 
Determining the set of transversal (or locality preserving) logic gates admitted by the boundary states of the \textbf{3F} Walker--Wang model remains an open problem. 
We remark that an extension of the Walker--Wang model has recently been defined which is capable of realizing an arbitrary symmetry-enriched topological order on the boundary under a global on-site symmetry action~\cite{bulmash2020absolute}. 
The full symmetry structure of an abelian Walker--Wang model with global symmetry should be captured by a 2-group~\cite{Benini2018}, we leave the investigation of MBQC with 2-group SPTs to future work.

Further interesting open directions include the construction of MBQC schemes using Walker--Wang resource states based on more exotic, non-abelian anyon theories, including those that are braiding universal, i.e., not requiring non-topological magic state preparation and distillation. 
Moving away from stabilizer resource states, it may be difficult to keep track of, and control, the randomness induced by the local measurements. One way to address this concern would be to consider adiabatic approaches to MBQC~\cite{bacon2013adiabatic,williamson2015symmetry} to circumvent some these difficulties.

Another interesting direction is to investigate MBQC schemes based on Walker--Wang resource states that are both perturbatively and thermally stable. The \textbf{3F} Walker--Wang model can be shown to belong to a nontrivial SPT phase under $\zz_2^2$ 1-form symmetries using the same arguments as Refs.~\cite{roberts2017symmetry,roberts2020symmetry}. More generally, the bulk of any Walker--Wang state corresponding to a modular anyon theory input should be a nontrivial SPT order under some 1-form symmetry (or an appropriate generalization thereof). 
The nontrivial nature of these 1-form SPT phases is manifest through their anomalous symmetry action on the boundary. This anomalous boundary action of the 1-form symmetry corresponds to the string operators of a modular anyon theory. A gapped phase supporting that anyon theory can be used to  realize a gapped boundary condition that fulfils the required anomaly matching condition. 
The topologically ordered boundaries of these states should remain thermally stable under the 1-form symmetries. Demonstrating the stability (or otherwise) of these schemes away from fixed point models is an open problem: the computation scheme is based on symmetry principles alone, and (potentially fattened) string operators and defects that exist throughout the topological phase should suffice to perform topological quantum computation. 

Finally, to complement the MBQC scheme, we suggested an alternative approach to TQC based on symmetry defects of the \textbf{3F} anyon theory realized in the 2D topological subsystem code due to Bomb\'{i}n supplemented with code deformations to implement braiding, following Ref.~\cite{Bombin_2011}. The 2-body nature of the gauge generators for the 2D subsystem code may be attractive for architectures with strong locality constraints or long two qubit gate times. Investigation into the error-correcting performance of the 2D topological subsystem code remains an important open problem in this direction.

\acknowledgements
We acknowledge inspiring discussions with Guanyu Zhu at the early stages of this project. 
We also acknowledge useful discussions with Jacob Bridgeman. 
DW acknowledges support from the Simons foundation.


%

\appendix

\onecolumngrid

\section{Verifying the defect computation scheme}\label{appProofOfGates}
In this section we prove Lemmas~\ref{lemSingleQubit}, \ref{lemEntangling}. Our proofs are applicable when the defect computation scheme is implemented by code deformation~\cite{bombin2009quantum}, or by MBQC as in Sec.~\ref{sec3FMBQC}. We prove that the braids induce the correct logical action by inferring their mapping on the Pauli logical operators. In the case of Clifford operators, we can uniquely identify a Clifford operation $C$ by its action on Pauli operators under conjugation,
\begin{equation}\label{eqCliffordConjugation}
    C: P \rightarrow CPC^{\dagger}, \quad P \in \mathbb{P}_n,
\end{equation}
where $\mathbb{P}_n$ is the Pauli group on $n$ qubits. Furthermore, it is sufficient to determine its action on Pauli $X$ and $Z$ logical operators. To prove that the twist braid or fusion induces the correct logical action, we must prove that logical operators for the input are mapped according to Eq.~(\ref{eqCliffordConjugation}) at the output of the channel. This is achieved by finding correlation surfaces that can be thought of as world-sheets of logical operators, reducing to the correct logical operators on the boundaries induced by measurement. Before continuing, we discuss how to understand the proofs in terms of code deformation and MBQC.

\subsubsection{Code deformation.} In the case of code-deformation, we take a codimension-1 foliation of the 3D braiding diagrams of Sec.~\ref{sec3FTQCScheme} to define a sequence of codes, one for each time-step. In other words, each equal-time slice $k$ of the defect braiding diagrams in Sec.~\ref{sec3FTQCScheme} represents a (subsystem stabilizer) code $\mathcal{G}_k$ with a given configuration of twist defects and domain walls. To implement the braid, one sequentially measures the operators from $\mathcal{G}_1, \mathcal{G}_2, \ldots, \mathcal{G}_m$, such that both $\mathcal{G}_1=\mathcal{G}_m$ are as per Fig.~\ref{figBaseEncodings}. Doing so, one builds up a history of measurement results that are used to decode. Up to the Pauli corrections that depend on decoder output, we can track the presence of a set of representative logical operators $\{\overline{X}^{(k)}_i, \overline{Z}^{(k)}_i\}_i$ for each time-slice $k$. For each gate, we illustrate how $\{\overline{X}^{(1)}_i, \overline{Z}^{(1)}_i\}_i$ propagate through the code deformation onto $\{C\overline{X}^{(m)}_iC^{\dagger}, C\overline{Z}^{(m)}_iC^{\dagger}\}_i$. 
We refer to the collection of logical operators that result as $\{\overline{X}^{(1)}_i, \overline{Z}^{(1)}_i\}_i$ are propagated through the code deformation as a correlation surface, to connect with the setting of MBQC. 

\subsubsection{Measurement-based quantum computation.} The case of MBQC is similar to code deformation. The tracking of how logical operators propagate through code deformation is replaced by the notion of correlation surfaces in MBQC, as defined in Sec.~\ref{sec3FMBQC}. In particular, the correlation surfaces must agree with the measurement pattern (i.e., their restriction to a single site must give the measurement observable), and the measurement outcomes along their support determine (along with the decoder output) the required Pauli correction. Up to this Pauli correction, they concretely represent how logical operators are propagated through space-time as measurements are implemented. Thus in MBQC, finding the correlation surfaces not only gives a proof that the gate has the correct action, but is required for its deterministic operation (i.e., to determine the logical Pauli frame).

The correlation surfaces are obtained for each preparation, gate, and measurement, and schematically depicted in Figs.~\ref{figHadamardCorrelations}-~\ref{figXZPrepCorrelations}. They are stabilizers of the resource state that can be viewed as topologically nontrivial 1-form operators that enclose (and measure) the charge on relevant sets of defects in the history state. To obtain the correlation surfaces for each operation, we discretize each surface $S$ of each gate in Figs.~\ref{figHadamardCorrelations},~\ref{figCZCorrelations} of App.~\ref{appProofOfGates} by defining a suitable 1-cocycle $c$ (i.e., a set of edges dual to a 2-cycle in the dual lattice). 
Such a 1-cocycle generally has some (co)boundary on the input or output layers (these compose with the correlation surfaces of subsequent gates).
Each edge of the correlation surface, $i \in c$, also carries a fermion label $l(i)$ from $\{\fr, \fg, \fb \}$ according to the logical operator being propagated, which can transform across a domain wall.  Then we define a correlation surface $S$ by
\begin{align}
    S = \prod_{i \in c} m_i, \quad \text{where} \quad  m_i = \begin{cases}
    \tau_i^X ~&\text{if}~ l(i) = \fr,\\
    \sigma_i^X ~&\text{if}~ l(i) = \fg, \\
    \sigma_i^X \tau_i^X ~&\text{if}~ l(i) = \fb. \\
    \end{cases}
\end{align}
For example, in the figures that follow, we schematically depict the correlation surface labels $l(i)= \fr$, $l(i)= \fg$ and $l(i)= \fb$ by red, green and blue, respectively. For magic state preparation the operators $m_i$ must be modified near the preparation site, much like the measurement group $\mathcal{M}$.

Importantly, correlation surfaces are not unique. Recall, for a given operation, we label the set of all correlation surfaces up to equivalence under $\mathcal{S}$ by $\overline{\mathcal{S}}$. In particular, the correlation surfaces in Figs.~\ref{figHadamardCorrelations}-\ref{figXZPrepCorrelations}, can be viewed as maps between pairs of representative logical operators. In general, one can construct multiple equivalent correlation surfaces between any pair of equivalent logical operator representatives, such as those shown in Fig.~\ref{fig3F2DEncodingsIsotope}. In Figs.~\ref{figHadamardCorrelations}-\ref{figXZPrepCorrelations}, we chose convenient representatives for illustration purposes. In order to compose correlation surfaces between consecutive channels, we choose a fixed set of logical representatives for the input and output of each channel.

\subsubsection{Proof of Lemma~\ref{lemSingleQubit}: the single qubit gates.}
We now prove that the correct action is induced under the $H$ and $S$ gates. Firstly, we recall the action of these gates under conjugation
\begin{align}
    H:& ~ X \mapsto Z, \\
    H:& ~ Z \mapsto X, \\
    S:& ~ X \mapsto Y, \\
    S:& ~Z \mapsto Z.
\end{align}
To prove that the braids have the correct logical action, we must find correlation surfaces that induce the correct action on logical operators. We diagrammatically present correlation surfaces that correctly propagate the respective logical operators are propagated for each gate in Fig.~\ref{figHadamardCorrelations}.

\begin{figure}[h]%
	\centering
	\includegraphics[width=0.15\linewidth]{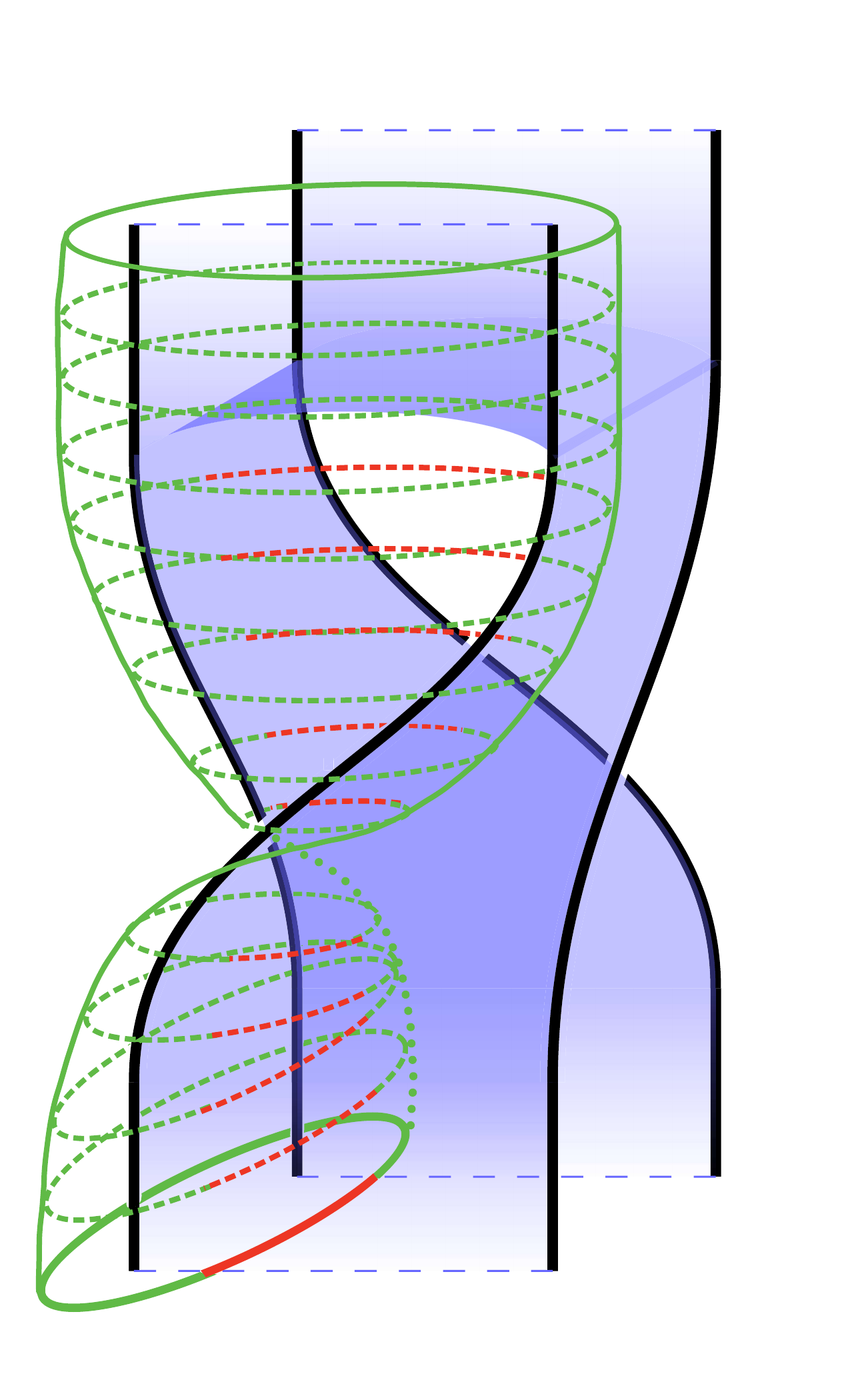} \qquad 
	\includegraphics[width=0.15\linewidth]{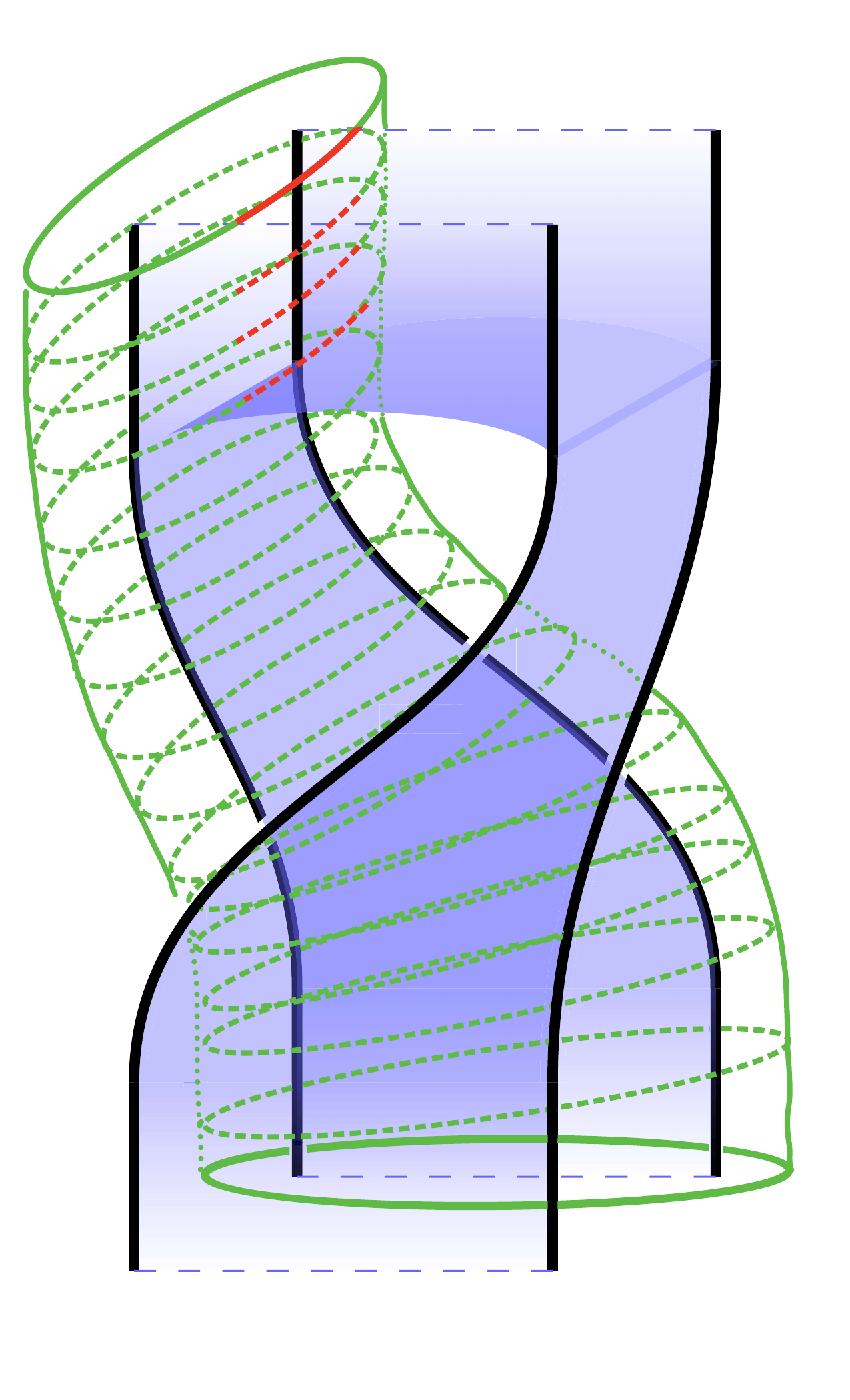} \qquad \qquad \qquad
	\includegraphics[width=0.15\linewidth]{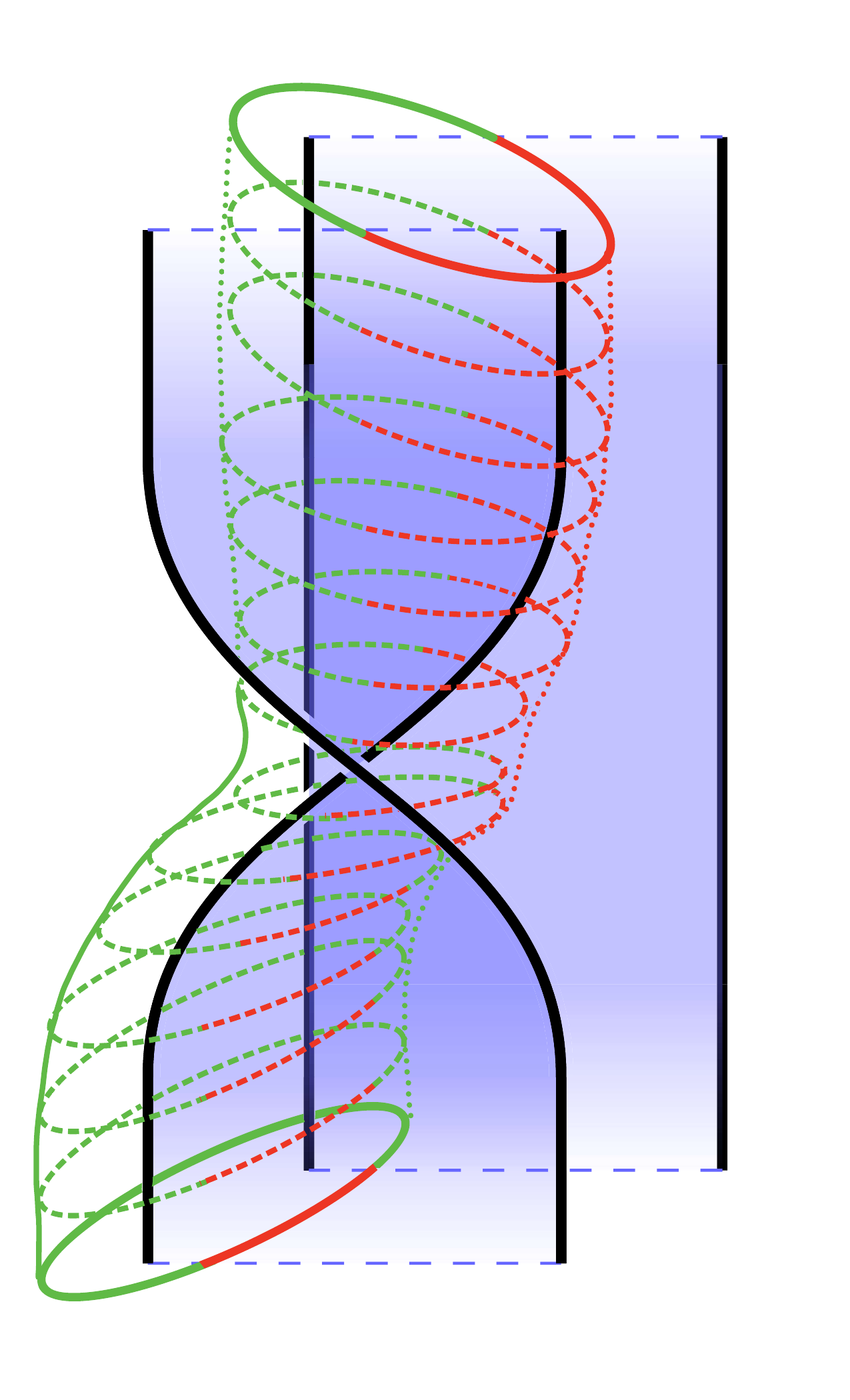} \qquad 
	\includegraphics[width=0.15\linewidth]{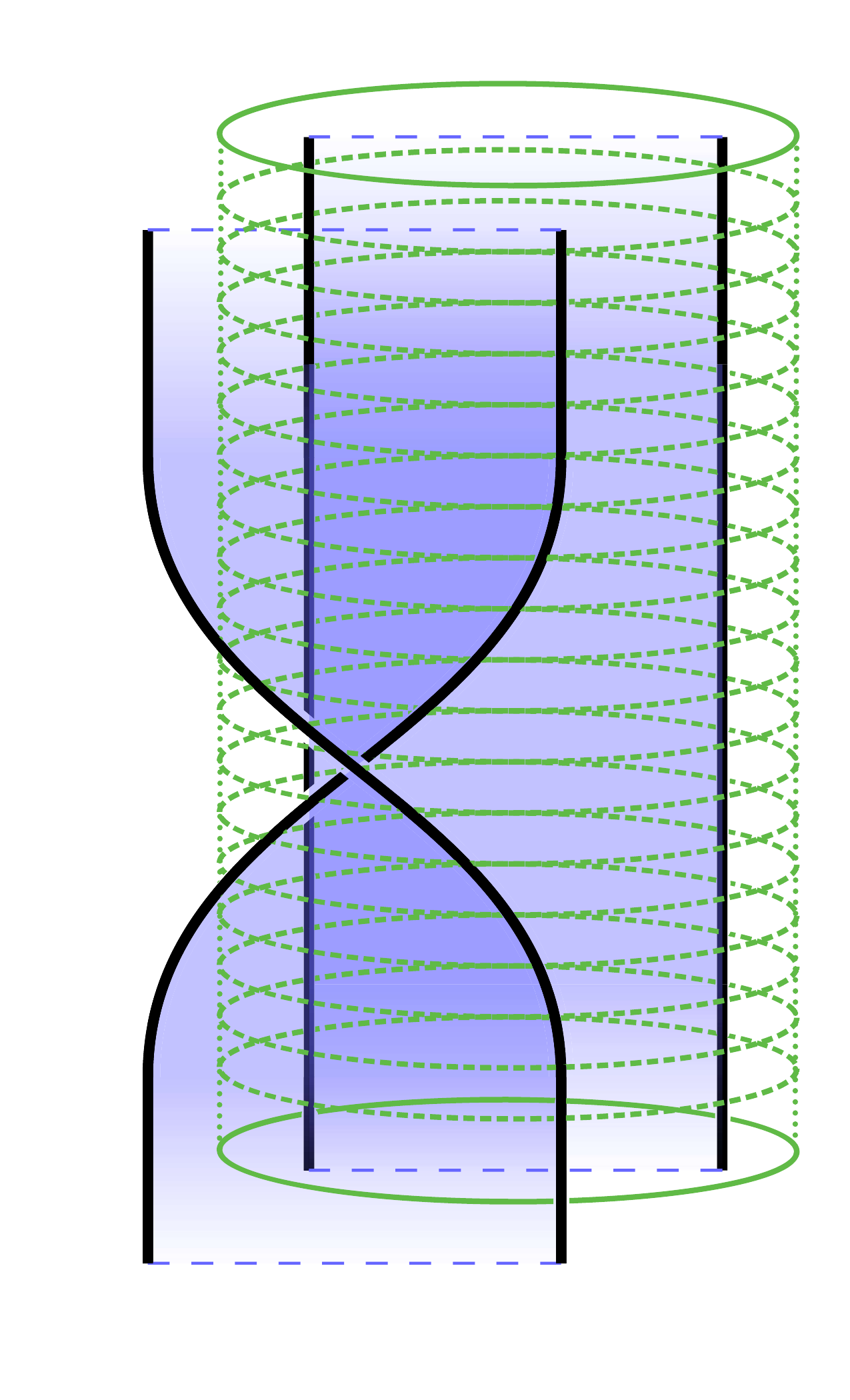} 
	\caption{Correlation surfaces for the Hadamard and Phase gate. Time moves upwards. (left) Correlation surfaces for $\overline{X} \rightarrow \overline{Z}$ and $\overline{Z} \rightarrow \overline{X}$ in the Hadamard gate. (right) Correlation surfaces for $\overline{X} \rightarrow \overline{Y}$ and $\overline{Z} \rightarrow \overline{Z}$ in the Phase gate.}
	\label{figHadamardCorrelations}
\end{figure}

\subsubsection{Proof of Lemma~\ref{lemEntangling}: the entangling gates.}
We first show that in the simplest case, when we use $g$, $h$-encodings with $g,h\in S_3$ distinct 2-cycles, we can achieve the Controlled-$Z$ gate $CZ_{1,2}$. Recall
\begin{align}
    CZ_{1,2}:&~ X_1 \mapsto X_1Z_2, \\
    CZ_{1,2}:&~ Z_1 \mapsto Z_1, \\
    CZ_{1,2}:&~ X_2 \mapsto X_2Z_1, \\
    CZ_{1,2}:&~ Z_2 \mapsto Z_2.
\end{align}

Similarly to the single qubit gates, we find correlation surfaces that propagate the logical operators according to the above equations. We diagrammatically represent these correlation surfaces in Fig.~\ref{figCZCorrelations}.

\begin{figure}[h]%
	\centering
	\includegraphics[width=0.22\linewidth]{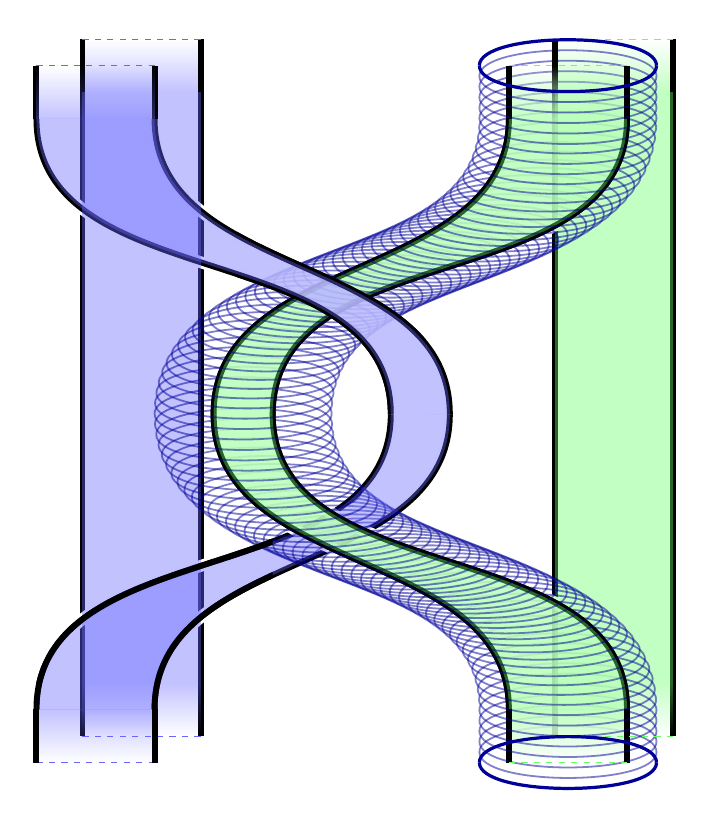} \qquad
	\includegraphics[width=0.22\linewidth]{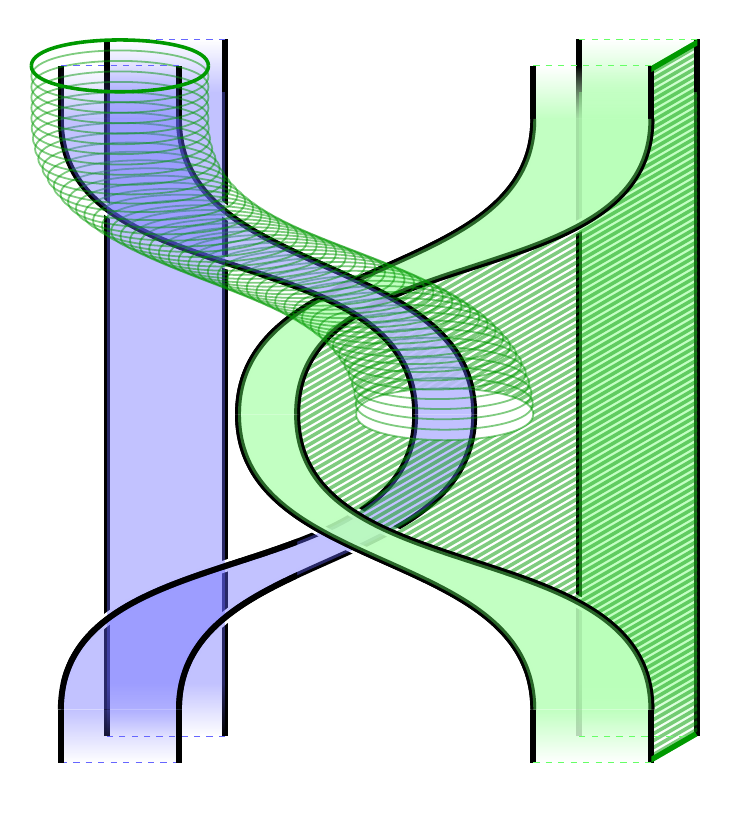} 	
	\caption{Correlation surfaces illustrating the action of the $CZ$ gate. Time moves upwards. We see $\overline{Z}_2\mapsto \overline{Z}_2$ (left) and $\overline{X}_2\mapsto \overline{X}_2\overline{Z}_1$ (right). One can similarly confirm $\overline{Z}_1\mapsto \overline{Z}_1$ and $\overline{X}_1\mapsto \overline{X}_1\overline{Z}_2$ using directly analogous correlation surfaces. This action on the 4 generating Pauli operators uniquely determines the $CZ$ gate.}
	\label{figCZCorrelations}
\end{figure}

We remark that the figures are unchanged when at least one of $g,h$ is a 3-cycle. Finally, we remark that if $g=h$ are both 2-cycles, then such a braid results in the identity. To see this, note that the fermion loop defining the logical operator on the output of the first $g$-encoding in Fig.~\ref{figCZCorrelations}~(right) can condense on the $\mathcal{T}_g$ twists and is invariant under the domain wall, and thus is precisely the identity logical operator (in other words, it can be isotoped to the vacuum configuration). Thus the fermionic loop operators defining logical $\overline{X}$ are transparent to the twist defects and domain walls and thus are unchanged under braiding.

\subsubsection{Proof of Prop.~\ref{prop3FCliffordUniversality}: Clifford universality.}
To complete the proof of Prop.~\ref{prop3FCliffordUniversality}, we need only show that the preparations and measurements work as required. To see this we show that eigenstates for either logical operators $\overline{X}$ or $\overline{Z}$ can be prepared exactly. The eigenvalue depends on the random measurement outcomes. We diagrammatically show this in Fig.~\ref{figXZPrepCorrelations} for the case of state preparations. To obtain measurements we time-reverse the diagram. One can similarly find correlation surfaces for the other $g$-encodings.

\begin{figure}[h]%
	\centering
	\includegraphics[width=0.4\linewidth]{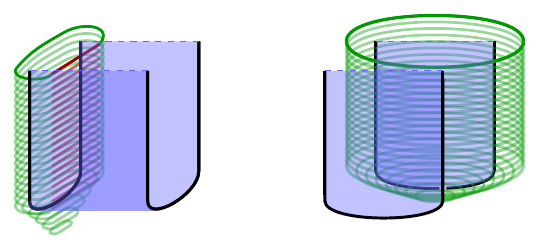} 
	\caption{Correlation surfaces for the $\overline{X}$ and $\overline{Z}$ preparations, time travels upwards. At the topmost timeslice, the state is an eigenstate of either logical $\overline{X}$ or $\overline{Z}$. Logical measurements are obtained by time-reversing the diagrams.}
	\label{figXZPrepCorrelations}
\end{figure}

\subsection{Majorana mapping for a single 2-cycle symmetry defect sector}\label{AppMajoranas}

We remark that for the theory of \textbf{3F} with a $\zz_2\leq S_3$ symmetry, generated by a single 2-cycle $g$, the twist defects can be represented by Majorana fermion operators, following a similar mapping for the toric code~\cite{bombin2010topologicaltwist}. 
This provides a concise description of the single qubit gates from the previous section and allows us to prove that entangling gates are not possible with a single type of symmetry defect using $g$-encodings. 
We remark that this mapping does not hold for the full $S_3$ defect theory, demonstrating that it is richer than the Ising defects in the toric code. 
The fact that all Clifford gates are achievable through braiding with the full \textbf{3F} defect theory is one of the key advantages that the richer symmetry group of this anyon theory provides. 

Following the surface code prescription~\cite{bombin2010topologicaltwist,barkeshli2013twist}, we represent each twist $\mathcal{T}_g^{(i)}$ by a Majorana operator $\gamma_i$. Such operators satisfy 
\begin{align}\label{eqMajoranaAntiCommutation}
   \gamma_j \gamma_k + \gamma_k \gamma_j = 2\delta_{kj},
\end{align}
where $\delta_{jk}$ is the Kronecker delta. 

We can represent the logical operators for $n$ qubits encoded in the $g$-encodings of Sec.~\ref{sec3FEncodings} in terms of Majorana modes as
\begin{align}
    \overline{X}_i &= \gamma_{1+4(i-1)} \gamma_{3+4(i-1)}, \label{eqMajoranaEncodingX} \\
    \overline{Z}_i &= \gamma_{1+4(i-1)} \gamma_{2+4(i-1)}, \label{eqMajoranaEncodingY}\\ 
    \overline{Y}_i &= \gamma_{1+4(i-1)} \gamma_{4+4(i-1)}, \label{eqMajoranaEncodingZ}
\end{align}
where the index on the Majorana operator labels the twist position as per Fig.~\ref{figQubitLabelling2}.
\begin{figure}[h]%
	\centering
	\includegraphics[width=0.2\linewidth]{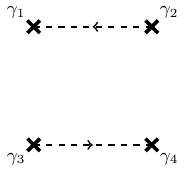} 
	\caption{Index convention for a single logical qubit.}
	\label{figQubitLabelling2}
\end{figure}

We remark that our choice to enforce trivial total charge in the  $g$-encoding leads to the redundancy condition $\gamma_1 \gamma_2 \gamma_3 \gamma_4 = I$. A braid is represented by an element $\sigma' \in B_{4n}$ the braid group on $4n$ strands. Importantly, up to a sign, only the induced permutation of the braid matters (i.e., we only care about $S_{4n} \cong B_{4n} / P_{4n}$, where $P_{4n}$ is the pure braid group), as any pure braid leads only to a Pauli operation.

Any braid $\sigma'$ giving rise to a permutation $\sigma \in S_{4n}$ induces the following logical action:
\begin{align}
    \overline{X}_i \mapsto \gamma_{\sigma(1+4(i-1))} \gamma_{\sigma(3+4(i-1))}, \\
    \overline{Z}_i \mapsto \gamma_{\sigma(1+4(i-1))} \gamma_{\sigma(2+4(i-1))},
\end{align}
where it is understood that $\sigma(i)$ is the permutation $\sigma$ applied to the index $i$. We observe that braids preserve the Majorana fermion operator weight and therefore can only map single qubit Pauli operators to single qubit Pauli operators -- meaning they belong to the set of single qubit Clifford operators. 
Secondly, we have verified that all single qubit Cliffords can be realized. 
It is sufficient to verify this on a single qubit. In particular,  $\sigma_H = (1342)$ generates the $H$ gate, and $\sigma_{S} = (34)$ generates the $S$ gate (both up to a Pauli operator):

\begin{align}
    \sigma_H : \overline{X} &\mapsto \gamma_{3}\gamma_{4} \sim \gamma_{1}\gamma_{2} = \overline{Z}, \\
    \sigma_H : \overline{Z} &\mapsto \gamma_{3} \gamma_{1} \sim \gamma_{1}\gamma_{3} = \overline{X},
\end{align}
and
\begin{align}
    \sigma_{S} : \overline{X} &\mapsto \gamma_{1}\gamma_{4} = \overline{Y}, \\
    \sigma_{S} : \overline{Z} &\mapsto \gamma_{1} \gamma_{2} = \overline{Z}.
\end{align}

\section{Topological compilation}\label{SecEntanglingCircuitViaAncilla}

We can utilize encodings in $\mathcal{T}_{(\text{rb})}$ twists as ancillas to mediate gates between logical qubits encoded in $\mathcal{T}_{\groupone}$ twists. The circuit in Fig.~\ref{fig3FEntanglingCircuit} implements a $CZ$ gate between two $\groupone$-encoded twists using an $(\text{rb})$-encoded ancilla. 
\begin{figure}[h]%
	\centering
	\includegraphics[width=0.3\linewidth]{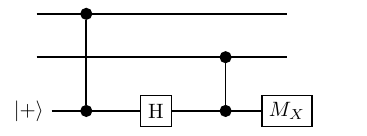} 
	\caption{Circuit that implements $CZ$ between the two qubits on the top wires. The circuit is composed of operations that are implementable via elementary gates outlined in the main text. Here $M_X$ is a Pauli-$X$ measurement. In particular this circuit is designed to implement a $CZ$ gate between two $\groupone$-encoded qubits (top two wires) using an $(\text{rb})$-encoded ancilla (bottom wire). }
	\label{fig3FEntanglingCircuit}
\end{figure}

We now describe the Clifford unitaries of Sec.~\ref{sec3FGates} in terms of elements of the braid group. Utilising the braid representation may allow one to topologically compile to find more efficient representatives of general elements of the multi-qubit Clifford group. With the labelling from Fig.~\ref{figBaseEncodings}, the Hadamard and $S$ gate are induced by the following braids (using standard braid group notation)
\begin{align}
    S: ~\sigma_3, \qquad H: ~\sigma_2 \sigma_3 \sigma_1,
\end{align}
as shown in Fig.~\ref{figSingleQubitBraids}. We remark that these are not the only braids that give rise to the required gate. 

We remark that if we only care about the single-qubit Clifford gates up to Pauli operators, then we can quotient by the pure braid group (i.e braids that do not permute the twists) as they only generate Pauli operations on qubits encoded in $\mathcal{E}_g$. Then doing so we can represent the Hadamard and $S$ gate up to Paulis with the following permutations: 
\begin{align}
    S: ~ (34), \qquad H:~ (1342).
\end{align}

The $CZ$ gate is also induced by a braid operation on 8 twists. In this case, the braid is a pure braid, meaning the permutation action is trivial. The braid element is depicted in Fig.~\ref{figTwoQubitBraids}.

\begin{figure}[h]%
	\centering
	\includegraphics[width=0.17\linewidth]{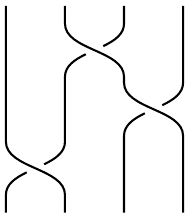} \hspace{2cm}
	\includegraphics[width=0.17\linewidth]{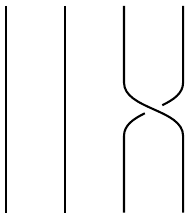} 
	\caption{Braid diagrams for the $H$ gate (left) and the $S$ gate (right). Strings denote the twists $\mathcal{T}_{\groupone}$ -- labelled 1-4 -- that must be braided. The twists are labelled from left to right $\mathcal{T}_{\groupone}^{(1)}, \mathcal{T}_{\groupone}^{(2)}, \mathcal{T}_{\groupone}^{(3)}, \mathcal{T}_{\groupone}^{(4)}$ according to Fig.~\ref{figBaseEncodings}.}
	\label{figSingleQubitBraids}
\end{figure}
\begin{figure}[h]%
	\centering
	\includegraphics[width=0.25\linewidth]{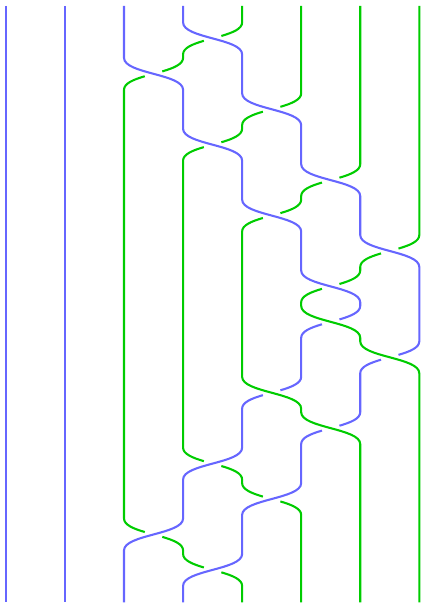} 
	\caption{Braid diagrams for the $CZ$ gate. Strings denote the twists $\mathcal{T}_{\groupone}$, $\mathcal{T}_{(\text{rb})}$ -- labelled 1-8 -- that must be braided. The first four represent the first encoded qubit and the second four represent the second encoded qubit. Note that from left to right, the twists are arranged $\mathcal{T}_{\groupone}^{(1)}, \mathcal{T}_{\groupone}^{(2)}, \mathcal{T}_{\groupone}^{(3)}, \mathcal{T}_{\groupone}^{(4)}, \mathcal{T}_{(\text{rb})}^{(1)}, \mathcal{T}_{(\text{rb})}^{(2)}, \mathcal{T}_{(\text{rb})}^{(3)}, \mathcal{T}_{(\text{rb})}^{(4)}$ according to the twist labelling of Fig.~\ref{figBaseEncodings}.}
	\label{figTwoQubitBraids}
\end{figure}
We remark that the relabelling all crossings from under to over leaves the gate invariant up to a Pauli. This can be verified by inspection of the logical operators from Eqs.(\ref{eqMajoranaEncodingX})-(\ref{eqMajoranaEncodingZ}) along with the anti-commutation relations of Eq.~(\ref{eqMajoranaAntiCommutation}).

\section{Equivalence between $g$-encodings}\label{appEquivalences}

In this section we show that the domain wall location is not important in the $g$-encoding -- only the location of the symmetry defects (twists) matter. As a concrete example, we consider mapping between the two configurations in Fig.~\ref{fig2DIsotopicEncodings}~(top). Recall that the symmetry action applied to a region creates a domain wall on its boundary. By applying a symmetry transformation to the convex hull of the four twists in Fig.~\ref{fig2DIsotopicEncodings}~(top), we can map between the seemingly distinct $g$-encodings. This can be verified explicitly in (2{+}1)D: applying a symmetry transformation in the plane leads to the space-time domain wall configuration of Fig.~\ref{fig2DIsotopicEncodings}~(bottom), which explicitly maps between the two encodings. Note that for defects based on a 2-cycle, the orientation does not matter. If instead one used a 3-cycle on would have to keep track of the orientation of the domain wall worldsheet in Fig.~\ref{fig2DIsotopicEncodings}~(bottom).

\begin{figure}[h]%
	\centering
	\includegraphics[width=0.28\linewidth]{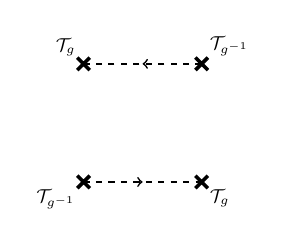} 
	\includegraphics[width=0.28\linewidth]{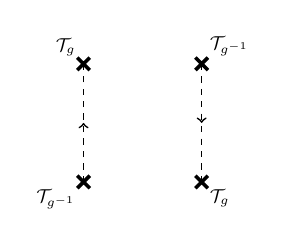} \\
	\includegraphics[width=0.18\linewidth]{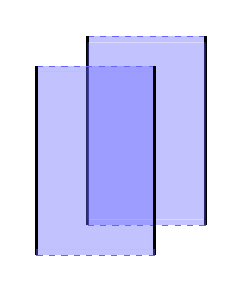} \quad
	\includegraphics[width=0.18\linewidth]{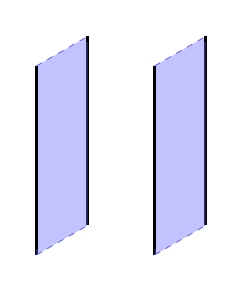} \quad
	\includegraphics[width=0.18\linewidth]{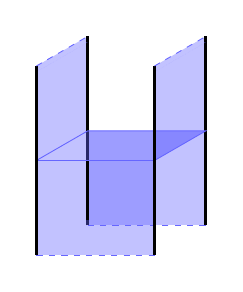} 	
	\caption{(top) Two equivalent representations of a $g$-encoding. The precise location of the domain walls does not matter, only their endpoint. (bottom) The two encodings are isotopic up to a space-like domain wall. The domain wall configration that achieves this is depicted on the right, where for simplicity we have ignored orientation }
	\label{fig2DIsotopicEncodings}
\end{figure}

\section{Trijunction encoding}\label{appOtherEncodings}
In this section we briefly review another defect encoding that may offer more efficient logic schemes. This encoding is called the trijunction encoding, and it utilizes a pair of defect triples. Each defect triple is as usual $g$-neutral. The only way to achieve $g$-neturality with three twists is if at least one of the defects is a 3-cycle defect. We illustrate the encodings and their logical operators in Fig.~\ref{figTrijunctionEncoding}.

\begin{figure}[h]%
	\centering
	\includegraphics[width=0.20\linewidth]{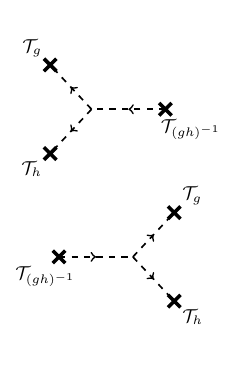} \qquad
	\includegraphics[width=0.20\linewidth]{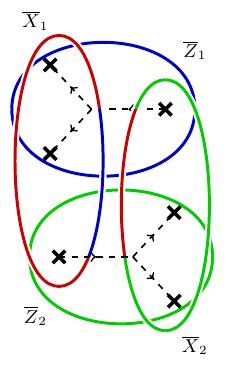} 
	\caption{A trijunction encoding. (left) Configuration of twist defects defined by $g,h\in S_3$. We must have either $g\neq h$, or $g=h$ and both $g,h$ are 3-cycles in order to have a valid tri-junction that encodes qubits: If we have exactly one of $g,h$ a 3-cycle then we encode one logical qubit, if both $g,h = \grouptwo$ or $\grouptwo^{-1}$ then we encode two qubits. (right) Trijunction encoding with $g= h = (gh)^{-1} = \grprgb$. Depicted are a set of fermionic Wilson loops that form a generating set of Logical Pauli operators.}
	\label{figTrijunctionEncoding}
\end{figure}

We remark that the trijunction defect encoding can be reduced to the standard $g$-encodings of Sec.~\ref{sec3FEncodings} by fusing one pair of defects from each trijunction. Nonetheless they may offer more efficient braiding schemes due to the larger possible set of braids.

\section{Proof that $U$ is a representation of $S_3$}\label{AppProofOfSymmetryRep}
In this section we verify Eqs.~(\ref{eqOnsiteSym1}),~(\ref{eqOnsiteSym2}) to show that $U$ is a representation of $S_3$. We need to verify Eq.~(\ref{eqFermionPermuation}) along with $g\cdot\ket{\textbf{1}}_i = \ket{\textbf{1}}_i$. Recall the local basis labelling $\ket{\textbf{1}} := \ket{++}$, $\ket{\fr} := \ket{-+}$, $\ket{\fg} := \ket{+-}$, $\ket{\fb} := \ket{--}$. Reducing this to the generators $S_3 = \langle \groupone, \grouptwo \rangle$, we need only show that that 
\begin{align}
    u_i(\grouptwoarg) &: \ket{-+}_i \mapsto \ket{+-}_i \mapsto \ket{--}_i \mapsto \ket{-+}_i, \\
    u_i(\grouponearg) &: \ket{+-}_i \leftrightarrow \ket{-+}_i.
\end{align}
One can directly compute that $u_i(\grouponearg) = \text{SWAP}_{i_1, i_2}$, $u_i(\grouptwoarg) = \text{SWAP}_{i_1, i_2} \cdot \text{CNOT}_{i_1, i_2}$ does the job.

\section{Proof of Prop.~\ref{prop3FHamSymmetry}}\label{AppProofOfProp3FHamSymmetry}
We now prove Prop.~\ref{prop3FHamSymmetry}. We show that the ground space of $H_{\textbf{3F}}$ is preserved by $S(g), g\in S_3$ by showing that the (stabilizer) group generated by the terms in $H_{\textbf{3F}}$ is invariant invariant under $S(\grouponearg)$ and $S(\grouptwoarg)$. Namely, let \begin{equation}
    \mathcal{R}_{\textbf{3F}} = \langle A_v^{(\fr)} , A_v^{(\fg)} , B_f^{(\fr)} , B_f^{(\fg)} ~|~ v\in V, f\in F \rangle. 
\end{equation}

Firstly, consider $S(\grouptwoarg) = U(\grouptwoarg)$. On each site $i$ we have
\begin{align}
    u_i(\grouptwoarg):& ~\sigma_i^X \rightarrow \sigma_i^X \tau_i^X, \\
    u_i(\grouptwoarg):& ~\sigma_i^Z \rightarrow \tau_i^Z, \\
    u_i(\grouptwoarg):& ~\tau_i^X \rightarrow \sigma_i^X,  \\
    u_i(\grouptwoarg):& ~\tau_i^Z \rightarrow \sigma_i^Z \tau_i^Z, \\
\end{align}
from which one can verify that 
\begin{align}
    U(\grouptwoarg):& ~B_f^{(\fr)} \rightarrow B_f^{(\fg)} \rightarrow B_f^{(\fr)}B_f^{(\fg)} \rightarrow B_f^{(\fr)},\\
    U(\grouptwoarg):& ~A_v^{(\fr)} \rightarrow A_v^{(\fr)} A_v^{(\fg)} \rightarrow A_v^{(\fg)} \rightarrow A_v^{(\fr)},
\end{align}
from which it follows $S(\grouptwoarg) \mathcal{R}_{\textbf{3F}} S(\grouptwoarg)^{\dagger} = \mathcal{R}_{\textbf{3F}}$.

Now for $S(\grouponearg) = V(\grouponearg)U(\grouponearg)$ with $U(\grouponearg) = \otimes_i\text{SWAP}_{i_1, i_2}$ and $V(\grouponearg) = T_{\tau}(1,1,1)$. On each site $i$ we have
\begin{align}
    u_i(\grouponearg): \sigma_i^{\alpha} \leftrightarrow \tau_i^{\alpha} \quad \alpha \in \{X, Y, Z\}
\end{align}
Therefore we have
\begin{align}
    U(\grouponearg):& ~A_v^{(\fr)} \leftrightarrow A_v^{(\fg)}, \\
    U(\grouponearg):& ~B_f^{(\fr)} \rightarrow \tau_{O_f}^X \tau_{U_f}^X \sigma_{U_f}^X \prod_{i \in \partial f}\tau_i^Z, \quad ~B_f^{(\fg)} \rightarrow \tau_{O_f}^X \sigma_{O_f}^X \sigma_{U_f}^X\prod_{i \in \partial f}\sigma_i^Z.
\end{align}
After translating all $\tau$ qubits in the $(1,1,1)$ direction, we see that in total
\begin{align}
    S(\grouponearg):& ~B_f^{(\fr)} \leftrightarrow B_f^{(\fg)},\\
    S(\grouponearg):& ~A_v^{(\fr)} \rightarrow A_{v'}^{(\fg)}, \quad A_v^{(\fg)} \rightarrow A_v^{(\fr)},
\end{align}
where $v'$ is the vertex obtained by shifting $v$ in the $(1,1,1)$ direction.
From which it follows $S(\grouponearg) \mathcal{R}_{\textbf{3F}} S(\grouponearg)^{\dagger} = \mathcal{R}_{\textbf{3F}}$.

\section{Symmetry defects in $H_{\textbf{3F}}$}\label{secWWSymmetryDefects}

\begin{figure}[t]%
    \centering	
    \includegraphics[width=0.49\linewidth]{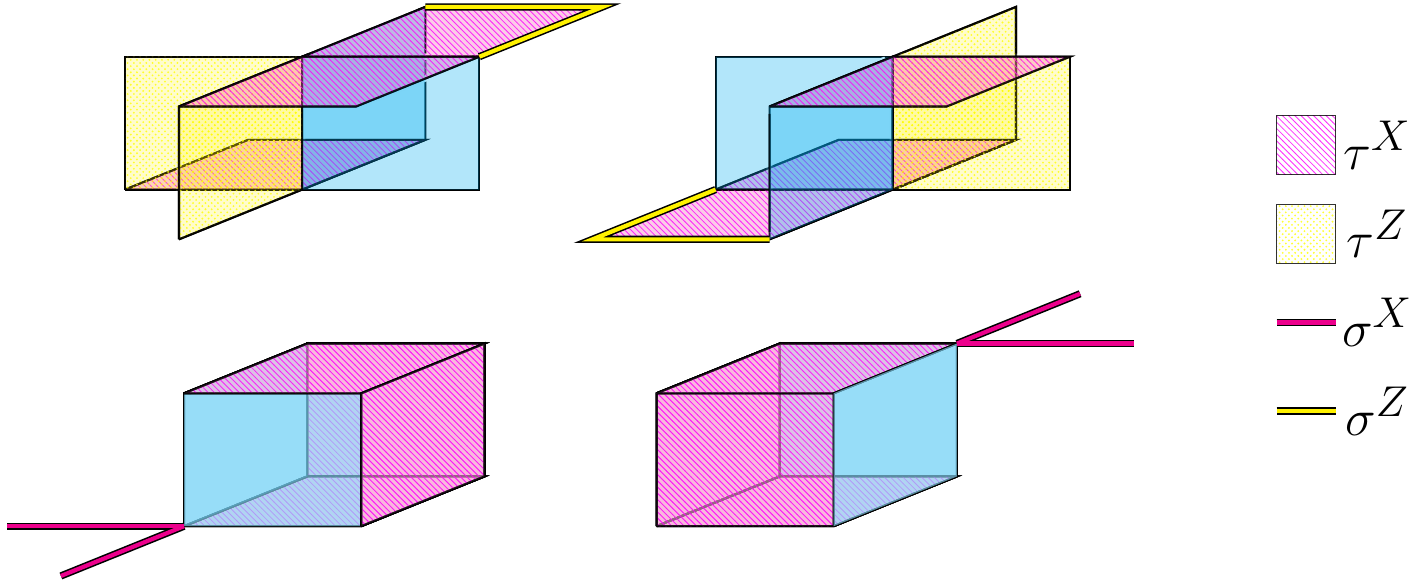} 
    \caption{Example \textbf{3F} Walker--Wang terms along the ${\groupone \in S_3}$ domain wall depicted in Fig. \ref{fig3FDomainWall}. The terms are color coded according to their support: blue shaded faces denote the domain wall $D$ -- upon which no qubits are supported; magenta shaded faces and edges denote the presence of $\tau^X$ and $\sigma^X$ respectively; yellow shaded faces and edges denote the presence of $\tau^Z$ and $\sigma^Z$ respectively. The top row of terms may be regarded as transformed versions of the right-most terms of Fig.~\ref{fig3FModifiedWWModel} that intersect the domain wall plane, while the bottom row are the transformed 1-form terms $\tilde{A}_v^{(\fr)}$ and $\tilde{A}_q^{(\fg)}$.}
    \label{fig3FDomainWallTerms}
\end{figure}

Here we compute \textbf{3F} Hamiltonian terms along domain walls and twists for the 2-cycle symmetry $\groupone \in S_3$. The domain walls and twists corresponding to 3-cycles $\grouptwo,(\text{rbg}) \in S_3$ are simple to construct, involving no change to the Hilbert space or lattice due to the onsite nature of their representation. The remaining symmetry defects and twists can be constructed by direct analogy, combining those of $\groupone$ with those of $\grouptwo$.

\subsection{\textbf{3F} domain walls and twists for $\groupone \in S_3$}

\begin{figure}[t]%
    \centering
    \includegraphics[width=0.49\linewidth]{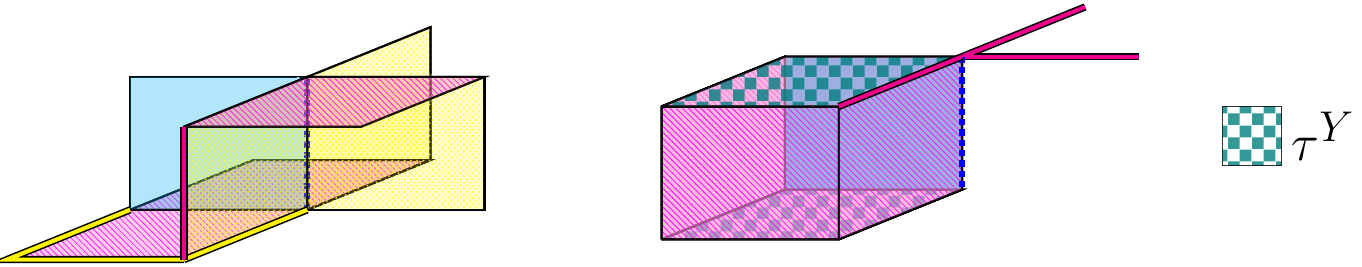}
    \caption{Example \textbf{3F} Walker--Wang terms along the ${\groupone \in S_3}$ twist depicted in Fig. \ref{fig3FDomainWall}. The twist travels along the central vertical edge, depicted by the dotted blue line. The left term can be regarded as the transformed version of right-most term of Fig.~\ref{fig3FDomainWallTerms} along the twist (obtained by multiplying a plaquette by its image under translation, each restricted to the qubits on the complement of the defect). Similarly, the modified 1-form operators contain $\tau^Y$ to ensure correct commutation. Other terms can be obtained by translating in the $\hat{y}$ direction, but we remark that these terms alone do not form a complete set. The color coding is identical to that of Fig.~\ref{fig3FDomainWallTerms} with the addition of $\tau^Y$ being denoted by chequered teal faces.}
    \label{fig3FTwists}
\end{figure}

For clarity, we consider the modified lattice with qubits on faces and edges, whose Hamiltonian terms are given by Eqs.~(\ref{eqModifiedHamTermse}),~(\ref{eqModifiedHamTermsm}). Consider a domain wall $D$ given by a plane normal to the $n = (1,0,-1)$ direction constructed using the translation symmetry ${S(\grouponearg) = T(w)}$, ${w=\frac{1}{2}(1,1,-1)}$ (both vectors were chosen for visualisation purposes). The discretised version of the domain wall on the lattice is visualised in Fig.~\ref{fig3FDomainWall}. 
The modified Hilbert space and Hamiltonian terms along the domain wall are depicted in Fig.~\ref{fig3FDomainWallTerms}. 
We remark that there is a layer of qubits missing on the domain wall itself, arising from the restricted translation action away from (rather than parallel to) the domain wall. 

Now consider a domain wall $D$ with boundary $\partial D$, for example along the $\hat{y} = (0,1,0)$, direction. 
To find a set of local terms to gap out the twist, we consider modifying the plaquette terms whose supports intersect the twist line to make them commute with the domain wall and bulk terms. 
The modified 1-form terms are then determined by constructing operators that commute with these modified plaquette terms and all other terms in the Hamiltonian (such that the product of modified and unmodified plaquette terms around a 3-, or 0-cell still matches the product of a pair of modified and unmodified 1-form terms). 
Such modifications can be done locally, following the discussion in Sec.~\ref{subsecTwistPrescription}. 
In Fig.~\ref{fig3FTwists}, we depict an example of a modified version of the right-most term from Fig.~\ref{fig3FDomainWallTerms} along the twist. 
We remark that the terms depicted in Fig.~\ref{fig3FTwists} alone do not form a complete set to gap out the twist.

One can also define these defects on the original Walker--Wang lattice, with two qubits per edge,  by applying the inverse transformation of Eq.~(\ref{eqLatticeTransformation}).

\end{document}